\documentclass{article}

\usepackage{arxiv}

\usepackage[utf8]{inputenc}
\usepackage{amssymb}
\usepackage{csquotes}
\usepackage{underscore}
\usepackage{graphicx}
\usepackage{xspace}
\usepackage{amsmath}
\usepackage{amsthm}
\usepackage{hyperref}
\newtheorem{definition}{Definition}
\newtheorem{proposition}{Proposition}

\newtheorem{example}{Example}

\hypersetup{%
    pdfcreator = {},%
    pdfproducer = {}%
}%

\usepackage{paralist}
\usepackage[shortlabels,inline]{enumitem}
\usepackage{algorithm}
\usepackage[noend]{algpseudocode}
\usepackage{eucal}
\def\sectionautorefname~{Sect.~}
\def\subsectionautorefname~{Sect.~}
\def\propositionautorefname~{Prop.~}

\def\figureautorefname~{Fig.~}
\def\algorithmautorefname~{Alg.~}
\def\exampleautorefname~{Ex.~}
\def\definitionautorefname~{Def.~}
\newcommand{\lref}[1]{L.~\ref{#1}}

\algrenewcommand\algorithmiccomment[2][\footnotesize]{{#1 \hfill\(\triangleright\)
\textit{#2}}}

\newcommand{\True}{\texttt{true}\xspace}

\algrenewcommand\algorithmicrequire{\textbf{Input:}}

\let\Until\undefined

\def\colorBright/{!40}
\def\colorPK/{orange}
\def\colorAR/{blue}
\def\colorMF/{yellow}
\def\colorWH/{green}
\def\colorWD/{red}

\newcommand{\carA}{A\xspace}
\newcommand{\carB}{B\xspace}

\providecommand{\liff}{\leftrightarrow}

\usepackage{wasysym}
\newcommand{\nec}{{\mathord\Box}}
\newcommand{\pos}{{\ensuremath{\mathord{\Diamond}}}}
\def\::{{:}}

\newcommand{\Globally}{\ensuremath{{\mathord\mathrm{G}}}}
\newcommand{\Next}{\ensuremath{{\mathord\mathrm{X}}}}
\newcommand{\Until}{\ensuremath{{\mathord\mathrm{U}}}}
\newcommand{\Past}{{\mathord\mathrm{P}}}
\newcommand{\Since}{{\mathord\mathrm{S}}}
\newcommand{\Finally}{\ensuremath{{\mathord\mathrm{F}}}}
\newcommand{\History}{{\mathord\mathrm{H}}}

\newcommand{\Axiom}[1]{\ensuremath{\mathbf{#1}}}
\newcommand{\System}[1]{\ensuremath{\mathrm{#1}}}

\newcommand{\StateSpace}{\ensuremath{S}}
\newcommand{\state}{\ensuremath{s}\xspace}
\newcommand{\Points}{\ensuremath{\mathcal{R}}}

\newcommand{\BeliefEntities}{\ensuremath{\mathcal E}}
\newcommand{\BeliefAtoms}{\ensuremath{\mathcal E_A}}
\newcommand{\Vars}{\ensuremath{\mathcal V}}
\newcommand{\Interp}{\ensuremath{\pi}}
\newcommand{\atom}{\ensuremath{e}\xspace}
\newcommand{\systemgraph}{justification graph\xspace}
\newcommand{\systemgraphs}{justification graphs\xspace}
\newcommand{\SystemGraph}{Justification Graph\xspace}
\newcommand{\SystemGraphs}{Justification Graphs\xspace}

\newcommand{\act}{\textit{act}}

\newcommand{\sysStates}{\StateSpace}
\newcommand{\world}{\ensuremath{M}\xspace}

\newcommand{\spred}{\ensuremath{\psi}}
\newcommand{\run}{\ensuremath{r}\xspace}
\newcommand{\proc}{\textit{A}\xspace}
\newcommand{\procB}{\textit{B}\xspace}
\newcommand{\Env}{\textit{Env}\xspace}
\newcommand{\Edg}{\textit{T}\xspace}
\newcommand{\history}{\ensuremath{\mathbb{H}}\xspace}
\newcommand{\Bel}{\ensuremath{\mathcal{B}}\xspace}
\newcommand{\Just}{\ensuremath{\mathbb{J}}\xspace}
\newcommand{\goalweight}{\ensuremath{\ensuremath{w}}}

\newcommand{\comp}{\textit{r}\xspace}

\newcommand{\Act}{\ensuremath{\mathcal{A}}\xspace}
\newcommand{\props}{\Vars\xspace}

\newcommand{\obs}{\textit{b}}
\newcommand{\strat}{\ensuremath{\delta}\xspace}
\newcommand{\labelE}{\ensuremath{\lambda}}
\newcommand{\labelS}{\Interp}

\newcommand{\goalp}{\ensuremath{\varphi}}
\newcommand{\WorldSet}{\ensuremath{\mathcal{M}}}
\def\conbel/{\ensuremath{\mathcal{M}}}
\def\bel/{\ensuremath{M}}
\def\worlditem/{\ensuremath{\baseinfoset/_M}}
\def\worlditemA/{\ensuremath{\worlditem/}}
\def\partialworlditem/{\ensuremath{\widetilde{M}}}
\def\worldset/{\ensuremath{\baseinfoset/_\mathcal{M}}}
\def\worldsetA/{\ensuremath{\worldset/}}
\def\po/{\ensuremath{\leq_{2}}}
\def\baseinfo/{\ensuremath{\sigma}}
\def\baseinfoset/{\ensuremath{\Sigma}}
\def\baseinfosetset/{\ensuremath{\zeta}}
\def\stratset/{\ensuremath{\Delta}}
\def\stratsetA/{\ensuremath{\stratset/_{A}}}
\def\stratsetB/{\ensuremath{\stratset/_{B}}}
\def\strategy/{\ensuremath{\strat{}}}
\def\strategyA/{\ensuremath{(\strategy/_{A}, \strategy/_{B}')}}
\def\strategyB/{\ensuremath{(\strategy/_{A}', \strategy/_{B})}}
\def\strategyAB/{\ensuremath{(\strategy/_{A}, \strategy/_{B})}}
\def\goalset/{\ensuremath{\GoalSet{}}}
\def\goalsetmax/{\ensuremath{\goalset/^{\mathit{max}}}}
\def\goalsetA/{\ensuremath{\goalset/_{A}}}
\def\goalsetB/{\ensuremath{\goalset/_{B}}}
\def\goalsetAmax/{\ensuremath{\goalset/_{A}^{\mathit{max}}}}
\def\goalsetBmax/{\ensuremath{\goalset/_{B}^{\mathit{max}}}}
\def\goal/{\ensuremath{\goalp{}}}
\def\vobst/{\ensuremath{v_{\text{obstacle}}}}
\def\vfast/{\ensuremath{v_{B,\text{fast}}}}
\def\vslow/{\ensuremath{v_{B,\text{slow}}}}
\def\pistate/{\ensuremath{\state{}}}
\def\obsset/{\ensuremath{\mathcal{O}}}
\def\obs/{\ensuremath{O}}
\def\ob/{\ensuremath{o}}
\def\piact/{\ensuremath{\Act{}}}
\def\piactA/{\ensuremath{\piact/_{A}}}
\def\piactB/{\ensuremath{\piact/_{B}}}
\def\pivars/{\ensuremath{\Vars{}}}
\def\pivarspart/{\ensuremath{V}}
\def\pivar/{\ensuremath{v}}
\def\pitransrel/{\ensuremath{E}}
\def\pitrans/{\ensuremath{e}}
\def\pilabele/{\ensuremath{\labelE{}}}
\def\seqB/{\ensuremath{\gamma}}
\def\action/{\ensuremath{\upsilon}}
\def\pijustset/{\ensuremath{E}}
\def\pijust/{\ensuremath{e}}
\def\lb{\ensuremath{\left\lbrace}}
\def\rb{\ensuremath{\right\rbrace}}
\def\sep/{\par\bigskip\hrule\par\bigskip}
\def\level/{\ensuremath{(C_{i})}}
\def\levelA/{\ensuremath{(C_{1})}}
\def\levelB/{\ensuremath{(C_{2})}}
\def\levelC/{\ensuremath{(C_{3})}}
\def\levelD/{\ensuremath{(C_{4})}}
\def\pibelentities/{\BeliefEntities{}}
\def\pibelentity/{\ensuremath{e}}
\def\pijusts/{\ensuremath{\Just{}}}
\def\time/{\ensuremath{t}}
\def\maxTime/{\ensuremath{T}}
\def\pihistset/{\ensuremath{\history{}}}
\def\pihist/{\ensuremath{h}}
\def\maxcw/{PossibleWorlds}
\def\resolve/{Resolve}
\def\maxcwt/{\textsc{\maxcw/}}
\def\resolvet/{\textsc{\resolve/}}
\def\conflict/{\ensuremath{\mathcal{C}}}

\def\pusht/{Push}
\def\popt/{Pop}
\def\findstratt/{FindStrategy}
\def\findstrattt/{\textsc{\findstratt/}}
\def\getjust/{GetJustifications}
\def\getjustt/{\textsc{\getjust/}}
\def\testt/{TestIfNotWinning}
\def\testtt/{\textsc{TestIfNotWinning}}
\def\fixt/{FixConflict}
\def\fixtt/{\textsc{FixConflict}}
\def\breakt/{\textbf{break}}

\def\stratfunctA/{StratA}
\def\stratfunctB/{StratB}
\def\stratfuncttA/{\textsc{\stratfunctA/}}
\def\stratfuncttB/{\textsc{\stratfunctA/}}

\newcommand{\pirun}[1]{\ensuremath{\func{r}{#1}}}

\newcommand{\func}[2]{\ensuremath{#1\!\left(#2\right)}}

\newcommand{\ag}{\ensuremath{e}\xspace}

\newcommand{\GoalSet}{\ensuremath{\Phi}\xspace}
\newcommand{\MaxGoalSet}{\ensuremath{\Phi^{\textsf{max}}}\xspace}
\newcommand{\myparagraph}[1]{\paragraph{#1.}}

\title{Justification Based Reasoning in Dynamic Conflict Resolution%
  \thanks{%
    This work is partly supported by the German Research
    Council (DFG) as part of the PIRE SD-SSCPS project
    (Science of Design of Societal Scale CPS,
    grant no. DA 206/11-1, FR 2715/4-1) and the
    Research Training Group SCARE
    (System Correctness under Adverse Conditions, grant no. DFG GRK 1765).
  }
}

\author{%
  Werner Damm
  \quad
  Martin Fr{\"a}nzle
  \quad
  Willem Hagemann
  \quad
  Paul Kr{\"o}ger
  \quad
  Astrid Rakow\\
  Department of Computing Science\\
  University of Oldenburg, Germany\\
  \texttt{\{werner.damm, martin.fraenzle, willem.hagemann, paul.kroeger, a.rakow\}@uol.de}
}

\begin{document}

\maketitle

\title{Justification Based Reasoning in Dynamic Conflict Resolution}
  
\begin{abstract}

    We study conflict situations that dynamically arise in traffic scenarios,
    where different agents try to achieve their set of goals and have to decide
    on what to do based on their local perception. We distinguish several types
    of conflicts for this
    setting. In order to enable modelling of conflict situations and the reasons
    for conflicts, we present a logical framework that adopts concepts from
    epistemic and modal logic, justification and temporal logic.  Using this
    framework, we illustrate how conflicts can be identified and how we
    derive a chain of justifications leading to this conflict. We discuss how
    conflict resolution can be done when a vehicle has local, incomplete
    information, vehicle to vehicle communication (V2V) and partially
    ordered goals.

\end{abstract}

\section{Introduction}
As humans are replaced by autonomous systems, such systems must be able to
interact with each other and resolve dynamically arising  conflicts. Examples of such conflicts arise when a car wants to enter the highway in dense traffic or simply when a car wants to drive faster than the preceding. 
Such \enquote{conflicts} are pervasive in road traffic and although traffic rules define a jurisdictional frame, the decision, e.g., to give way, is not uniquely determined but influenced by a list of prioritised goals of each system and the personal preferences of its user. 
If it is impossible to achieve all goals simultaneously, autonomous driving systems (ADSs) have to decide \enquote{who} will \enquote{sacrifice} what goal in order to decide on their manoeuvres. 
Matters get even more complicated when we take into account that the ADS has only partial information. 
It perceives the world via sensors of limited reach and precision. Moreover, measurements can be contradicting. An ADS might use V2V to retrieve more information about the world, but it inevitably has a confined insight to other traffic participants and its environment. 
Nevertheless, for the acceptance of ADSs, it is imperative to implement conflict resolution mechanisms that take into account the high dimensionality of decision making. These decisions have to be explained and in case of an incident, the system's decisions have to be accountable.

In this paper we study conflict situations as dynamically occurring in road traffic and develop a formal notion of conflict between two agents. 
We distinguish several types of conflicts and propose a conflict resolution process where the different kinds of conflicts are resolved in an incremental fashion. 
This process successively increases the required cooperation and decreases the privacy of the agents, finally negotiating which goals of the two agents have to be sacrificed. 
We present a logical framework enabling the analysis of conflicts.
This framework borrows from epistemic and modal logic in order to accommodate the bookkeeping of evidences used during a decision process. The framework in particular provides a mean to summarise consistent evidences and keep them apart from inconsistent evidences. We hence can, e.g., fuse compatible perceptions into a belief $b$ about the world and fuse another set of compatible perceptions to a belief $b'$ and model decisions that take into account that $b$  might contradict $b'$.
Using the framework we illustrate how conflicts can be explained and algorithmically analysed as required for our conflict resolution process. Finally we report on a small case study using a prototype implementation (employing the Yices SMT solver \cite{Yices}) of the conflict resolution algorithm.

\myparagraph{Outline}
In \autoref{sec:conflict} we introduce the types of conflict on a running example and develop a formal notion of conflict between two agents.
We elaborate on the logical foundations for modelling and analysing conflicts and the logical framework itself in \autoref{sec:justifications}. 
We sketch our case study on conflict analysis in \autoref{sec:casestudy} and outline in \autoref{sec:overview} an algorithm for analysing conflict situations as requested by our resolution protocol and for deriving explanation of the conflict for the resolution.
Before drawing the conclusions in \autoref{sec:conclusion}, we discuss related work in
\autoref{sec:related}.

\section{Conflict}\label{sec:section2}\label{sec:conflict}
 Already in 1969 in the paper \enquote{Violence, Peace and Peace Research} \cite{Galtung} J. Galtung presents his theory of the \emph{Conflict Triangle}, a framework used in the study of peace and conflict. Following this theory a conflict comprises three aspects:
 opposing \emph{actions}, incompatible \emph{goals}, inconsistent \emph{beliefs} (regarding the reasons of the conflict, knowledge of the conflict parties,\ldots). 
 
We focus on conflicts  that arise dynamically between two agents in road traffic. 
 We develop a characterisation of \emph{conflict} as a situation where one agent can accomplish its goals with the help of the other, but both agents cannot accomplish all their goals simultaneously and the agents  have to decide what to do based on their local beliefs.
In \autoref{formal} we formalise our notion of conflict. 
For two agents with complete information, we may characterise a conflict as:
Agents \proc and \procB are in conflict, if
\begin{enumerate*}[nosep]
	\item \proc would accomplish its set of goals $\GoalSet_\proc$, if \procB will do what \proc requests, while
	\item \procB would accomplish its set of goals $\GoalSet_\procB$, if \proc will do what \procB requests, and
	\item it is impossible to accomplish the set of goals $\GoalSet_\proc\cup \GoalSet_\procB$.
	\end{enumerate*}
	A situation where \proc and \procB both compete to consume the same resource is thus an example of a conflict situation.
	Since we study conflicts from the view-point of an agent's beliefs, we also consider believed conflicts, which can  be resolved by sharing information regarding the others observations, strategies or goals. 
%
%
%
To resolve a conflict we propose a sequence of steps  that require an increasing level of cooperation and decreasing level of privacy -- the steps require to reveal information or to constrain acting options. 
Our resolution process defines the following steps:
\begin{itemize}[nosep]
    \item[\levelA/]  Shared situational awareness
    \item[\levelB/]  Sharing strategies
    \item[\levelC/]  Sharing goals
    \item[\levelD/]  Agreeing on which goals to sacrifice and which strategy to follow	    
\end{itemize}
Corresponding to \levelA/ to \levelD/, we introduce different kinds of conflicts on a running example -- a two lane highway, where one car, \carA, is heading towards an obstacle at its lane and at the lane to its left a fast car, \carB, is approaching from behind (cf. \autoref{fig:conf}). 
An agent has a prioritised list of goals (like 1. \enquote{collision-freedom}, 2.\enquote{changing lane} and 3. \enquote{driving fast}). We assume that an agent's goals are achievable. 
\begin{figure}
	\begin{minipage}{0.5\textwidth}
		\centering
\includegraphics[width=0.65\textwidth]{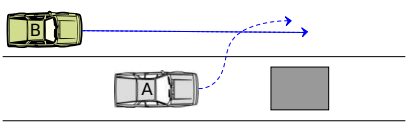}
\end{minipage}\hfill
\begin{minipage}{0.4\textwidth}
\caption{Car \proc wants to circumvent the obstacle (grey box). Car \carB is approaching from behind.}\label{fig:conf}
\end{minipage}
\end{figure}

An agent \proc has a set of \emph{actions} $\act_\proc$ and exists within a world. At a time the world has a certain state. The world \enquote{evolves} (changes state) as determined by the chosen actions of the agents within the world and events determined by the environment within the world.  
The agent perceives the world only via a set of \emph{observation predicates}, that are predicates whose valuation is determined by an observation of the agent. 
Without an observation the agent has no (direct) evidence for the valuation of the respective observation predicate.
\begin{example}\label{ex:evidence}
	Let car \proc want to change lane. It perceives that it is on a
	two lane highway, the way ahead is free for the next 500\,m and
	\procB is approaching. Let \proc perceive \procB's speed via
	radar. That is \proc makes the observation  \emph{\texttt{car\,\procB is fast}}
	justified by the evidence \emph{\texttt{radar}}.
	We annotate this briefly as  \emph{\texttt{radar:car\,\procB is fast}}.
	Further let \proc derive from lidar data that \procB is slow --
	\emph{\texttt{lidar:car\,\procB is slow}}.
\end{example}
In this situation we say agent \proc has \emph{contradicting evidences}.
Certain evidences can be combined without contradiction and others not. 
We assume that an agent organises its evidences in maximal consistent sets (i.e., \emph{\systemgraphs} of \autoref{sec:justifications}), where
each represents a set of possible worlds:
\begin{example}\label{ex:worlds}
	There are possible worlds of \proc where it is on a  two lane highway, the way ahead is free for the next 500\,m and  \procB is slowly approaching. Analogously \proc  considers possible worlds where \procB is fast. The state of the world outside of its sensors' reach is unconstrained.
\end{example}
Observing the world (for some time), an agent \proc assesses what it can do to achieve its goals in all possible worlds. That is, \proc tries to find a \emph{strategy}  that guarantees to achieve its goals in all its possible worlds. 
A strategy determines at each state the action of the agent -- the agent decides for an action based on its beliefs formed in the past regarding its possible worlds.
If there is one such strategy for \proc to accomplish its goals $\GoalSet_\proc$, then \proc has a (believed) winning strategy for $\GoalSet_\proc$. This strategy might not be winning in the "real" world though, e.g., due to misperceptions. 
\begin{example}\label{ex:winning}
	Let \proc want to drive slowly and comfortably. \proc wants to avoid collisions and it assumes that also \procB wants to avoid collisions. 
	Although \proc has contradicting evidences on the speed of \procB and hence believes that it is possible that \enquote{\procB is fast} and also that \enquote{\procB is slow}, it can follow the strategy to stay at its lane and wait until \procB has passed. This strategy is winning in all of \proc's possible worlds. 
\end{example}

Even when \proc has no believed winning strategy, it can have a winning strategy for a subset of possible worlds. Additional information on the state of world might resolve the conflict by eliminating possible worlds. We call such conflicts \emph{observation-resolvable} conflicts.
\begin{example}\label{ex:information}
	Let \proc want to change lane to circumvent the obstacle. It is happy to change directly after \procB but only if \procB is fast. 
	If \procB is slow, it prefers to change before \procB passed. 
	Further let \proc have contradicting evidences on the speed of \procB. 
	\proc considers a conflict with \procB possible in some world and hence has no believed winning strategy. Now it has to resolve its inconsistent beliefs.
	Let \procB tell \proc, it is fast, and \proc trust \procB more than its own sensors, then \proc might update its beliefs by dismissing all worlds where \procB is slow. Then \enquote{changing after \procB passed} becomes a believed winning strategy. 
\end{example}
In case of inconsistent evidences, as above, \proc has to decide how to update its beliefs. The decision how to update its beliefs will be based on the analysis of justifications (cf. \autoref{sec:justifications}) of (contradicting) evidences. The \texttt{lidar} contradicts the \texttt{radar} and \texttt{\procB} reports on its speed. Facing the contradiction of evidences justified by \texttt{lidar} and \texttt{radar} \proc trusts the evidence justified by \procB. 

Let the agents already have exchanged observations and \proc still have no believed winning strategy. A conflict might be resolved by communicating part of the other agent's (future) strategy:
\begin{example}\label{ex:strategy}
	Let \proc want to change lane. 
	It prefers to change directly after \procB, if \procB passes \proc fast. 
	Otherwise, \proc wants to change in front of \procB.
	Let \procB so far away that \procB might decelerate, in which case it might slow down so heavily that \proc would like to change in front of \procB even if \procB  currently is fast.

	Let \proc believe \enquote{\procB is fast}. Now \proc has no believed winning strategy, as \procB might decelerate. 
	According to (C2), information about parts of the agent's strategies are now communicated. \proc asks \procB whether it plans to decelerate. Let \procB be cooperative and tell \proc that it will not decelerate. Then \proc can dismiss  all worlds where \procB slows down and \enquote{changing after \procB passed} becomes a believed winning strategy for \proc.
\end{example}
Let the two agents have performed steps \levelA/ and \levelB/, i.e., they exchanged missing observations and strategy parts, and still \proc has no winning strategy for all possible worlds.
\begin{example}\label{ex:goals}
	Let now, in contrast to \autoref{ex:strategy}, \procB not tell \proc whether it will decelerate. Then step (C3) is performed. So \proc asks \procB to respect \proc's goals. Since \proc prefers \procB to be fast and \procB agrees to adopt \proc's goal as its own,  \proc can again dismiss  all worlds where \procB slows down. 
\end{example}
Here the conflict is resolved by communicating goals and the agreement to adopt the other's goals. So an agent's strategy might change in order to support the other agent.  We call this kind of conflicts \emph{goal-disclosure-resolvable} conflicts.

The above considered conflicts can be resolved by some kind of information exchange between the two agents, so that the sets of an agent's possible worlds is adapted and in the end all goals $\GoalSet_\proc$ of \proc and $\GoalSet_\procB$ of \procB are achievable in all remaining possible worlds. The price to pay for conflict resolution is that the agents will have to reveal information.
Still there are cases where simply not all goals are (believed to be) achievable. 
In this case \proc and \procB have to negotiate which goals $\GoalSet_{AB}\subseteq \GoalSet_\proc\cup \GoalSet_\procB$ shall be accomplished. 
While some goals may be compatible, other goals are conflicting. 
We hence consider goal subsets $\GoalSet_{\proc\procB}$ of $\GoalSet_\proc\cup \GoalSet_\procB$ for which a combined winning strategy for \proc and \procB exists to achieve $G_{\proc\procB}$.
We assume that there is a weight assignment function $w$ that assigns a value to a given goal combination $2^{\GoalSet_\proc\cup \GoalSet_\procB}\rightarrow \mathbb{N}$ based on which decision for a certain goal combination is taken.
This weighting of goals reflects the relative value of goals for the individual agents. 
Such a function will have to reflect, e.g., moral, ethics and jurisdiction.
\begin{example}\label{ex:neg}
	Let \proc's and \procB's highest priority goal be collision-freedom, reflected in goals $\varphi_{\proc,col}$ and $\varphi_{\procB,col}$. 
	Further let \proc want to go fast $\varphi_{\proc,fast}$ and change lane immediately $\varphi_{\proc,lc}$. 
	Let also \procB want to go fast $\varphi_{\procB,fast}$, so that \proc cannot change immediately.
	Now in step (C4) \proc and \procB negotiate what goals shall be accomplished. 
	In our scenario collision-freedom is valued most, and \procB's goals get priority over \proc's, since \procB is on the fast lane. Hence our resolution is to agree on a strategy accomplishing $\{\varphi_{\proc,col},\varphi_{\procB,col}, \varphi_{\procB,fast}\}$, which is  the set of goals  having the highest value among all those for which a combined winning strategy exists.
\end{example}

Note that  additional agents are captured as part of the environment here.
At each step an agent can also decide to negotiate with some other agent than \procB in order to resolve its conflict.

\subsection{Formal Notions}\label{formal}
In the following we introduce basic notions to define a conflict. 
Conflicts, as introduced above, arise in a wide variety of system models, but we consider in this paper only a propositional setting. 

Let $f_1:X\to Y_1$,\dots, $f_n:X\to Y_n$, and $f:X\to
Y_1\times\ldots\times Y_n$ be functions. We will write
$f=(f_1,\ldots,f_n)$ if and only if $f(x)=(f_1(x),\ldots,f_n(x))$ for
all $x\in X$. Note that for any given $f$ as above the decomposition into
its components $f_i$ is
  uniquely determined by the projections of $f$ onto the corresponding
codomain.

Each agent $\proc$ has a set of actions $\Act_\proc$. 
The sets of actions of two agents are disjoint.
To formally define a (possible)  world model of an agent $\proc$,
let $\sysStates$ be a set of states and $\props$ be a set of propositional variables.
$\props$ represents the set of belief propositions.
A state $\state\in\sysStates$ of a (possible) world is labelled with a subset $V\subset\props$ that is (assumed to be) true. $\props\setminus V$ is (assumed to be) false.

A (possible) world model $\world$ for an agent $\proc$ 
is a transition system over $\sysStates$ with designated initial state
and current state, all states are labelled with the belief propositions 
that hold at that state and transitions labelled with actions 
$\langle \act_\proc, \act_\procB, \act_\Env \rangle$ with $\act_\proc\in
\Act_\proc$ an action of agent \proc, $\act_\procB\in \Act_\procB$ an action
of agent \procB and $\act_\Env\in\Act_\Env$ an action of the environment.

The set of actions of an agent includes send and receive actions via which 
information can be exchanged, the environment guarantees to transmit a send 
message to the respective receiver.
Formally a possible world is $\world_{\proc} = (\sysStates, \Edg, \labelE,\labelS, \state_*, \state_c)$
with 
\begin{itemize}
	\item $\Edg \subseteq  \sysStates \times \sysStates$, 
	\item $\labelE: \Edg \rightarrow \Act_\proc \times \Act_\procB \times  \Act_\Env$,
	\item $\labelS: \sysStates \rightarrow 2^{\props}$,  
	\item $\state_* \in \sysStates$,  $\state_c \in \sysStates$ s.t.
		\begin{itemize}
			\item $\forall \state\in\sysStates$: $(\state,\state_*)\not\in \Edg$\hfill ($\state_*$ is the initial state)
			\item $\exists \state_1,\ldots,\state_{n+1} \in\sysStates$:$\forall 1\leq i \leq n$ $(\state_i,\state_{i+1})\in \Edg$,  $\land$ $\state_1=\state_*$ $\land$ $\state_{n+1}=\state_c$\\\null\hfill (the current state $\state_c$ is reachable from the initial state)\\ 
				$\land$	$((\state_i,\state') \in \Edg$ $\Rightarrow$ $\state'=\state_{i+1})$\hfill ( \world is linear between $\state_*$ and $\state_c$)
		\end{itemize}
	\end{itemize}

	The part of \world between $\state_*$ and $\state_c$ represents the \emph{history} of the current state. A finite \emph{run} in \world is a sequence of states $\run=\state_1\state_2\ldots\state_{n+1}$ with $\forall 1 \leq i\leq n: (s_i,s_{i+1})\in\Edg$.
	
There is one \enquote{special} world model that represents the ground truth, i.e.,  it reflects how the reality evolves.
An agent $\proc$ considers several worlds possible at a time.
This is, at each state \state of the real world, $\proc$ has a set of possible worlds $\WorldSet_\proc(\state)$.
The real world changes states according to the actions of $\proc$, $\procB$ and $\Env$. 
The set of possible worlds $\WorldSet_\proc(\state)$ changes to  $\WorldSet_\proc(\state')$ due to the passing of time and due to belief updates triggered by  e.g. observations. 
For the scope of this paper though, we do not consider the actual passing of time, but study the conflict analysis at a single state of the real world from the point of view of an agent.
Since at each state an agent \proc may consider several worlds possible, it may also consider several histories possible. 
A strategy is hence a function $\strat_{\proc}:2^{(2^\props)^{*}} \rightarrow \Act_\proc$, that determines an action for \proc based on the set of possible histories. 
$\mathbb{H}\in 2^{(2^\props)^{*}}$ represents a set of histories, where a history $h\in\mathbb{H}$ is given via the sequence of valuations of $\props$ along the path from $s_{*}$ to $s_{c}$. 
The set of possible histories at state \state  is the union of histories of possible worlds $\world_{\proc}\in\WorldSet_\proc(\state)$, denoted as $\mathbb{H}(\WorldSet_\proc(\state))$. 

Let $\run=\state_0\state_1\ldots\state_n$ be a run in $\world_\proc$ and $\gamma=\upsilon_1\upsilon_2\ldots\upsilon_n \in (\Act_\procB\times\Act_\Env)^n$ be a sequence of actions of agent $\procB$ and $\Env$ along $\run$. 
$\run$ follows strategy {\strat}, i.e., $\run=\comp(\strat, \gamma)$,  if $\labelE(s_{i-1},s_{i})=(\strat(\history(\WorldSet_\proc(\state_{i-1})),\upsilon_i),\forall 0\leq i\leq n$. We also write $\comp(\strat,\world_\proc)$ to denote the set of runs of $\world_\proc$ that follow \strat. 
We use linear-time temporal logic (LTL) to specify goals (cf. Def.~\autoref{def:model}). 
For a run $\comp$ and a goal (or a conjunction of goals) $\goalp$, we
write  $\comp \models \goalp$ if $\comp$ satisfies
$\goalp$, i.e., the valuation of propositions along the state sequence satisfies $\goalp$.\footnote{We assume that runs are infinite here. In case of finite runs, we make them infinite by repeating the last state infinitely often.}
We say \strat is a (believed) winning strategy for $\goalp$ in
$\world_{\proc}$, if all runs $\run$ of $\world_\proc$ that follow \strat also satisfy $\goalp$, $\forall \run\in\comp(\world_\proc,\strat):\run\models\goalp$.  We say that $\strat$ is a \emph{(believed) winning strategy} of \proc for $\goalp$ at the real world state \state if $\strat$ is a winning strategy for $\goalp$ in all possible worlds $\world_\proc\in\WorldSet_\proc(\state)$.

An agent $\proc$ has a set of goals $\GoalSet$ and a weight
assignment function $\goalweight_\proc:2^\GoalSet\rightarrow\mathbb{N}$ that
assigns values to a given goal combination. 
We write  $\comp\models \GoalSet$ as shorthand for
$\comp\models\bigwedge_{\goalp\in\GoalSet}\goalp$.
We  say $\GoalSet'\subset\GoalSet$ is a \emph{believed achievable goal} at real world \state if there is a strategy \strat, that is winning for the conjunction of all goals $\goalp\in\GoalSet'$ in all possible worlds $\world_\proc\in\WorldSet_\proc(\state)$.  We say $\GoalSet'\subset\GoalSet$ is a \emph{believed maximal goal} at real world state \state if its is a believed achievable goal and for all believed achievable goals $\GoalSet''\subset\GoalSet$ it holds that $w(\GoalSet')\geq w(\GoalSet'')$. 
The empty subgoal is defined to be true ($\top$).
For each world possible $\world_\proc\in\WorldSet_\proc(\state)$ agent \proc also has
\begin{compactenum}
	\item beliefs on the goals of $\procB$, $\GoalSet_\procB(\world_\proc)$, and 
	\item beliefs on the importance of subgoals of $\GoalSet_\procB(\world_\proc)$ to $\procB$, $\goalweight_\procB(\world_\proc)$, and 
	\item a set $\Just(\world_\proc)$ of justifications  for $\world_\proc$, $\GoalSet_\procB(\world_\proc)$ and $\goalweight_\procB(\world_\proc)$. 
\end{compactenum}
So at state \state of the real world an agent \proc has belief $\Bel(\proc,\state)=\bigcup_{\world_\proc\in\WorldSet_\proc(\state)}(\world_\proc,\GoalSet_\procB(\world_\proc), \goalweight_\procB(\world_\proc), \Just(\world_\proc))$.
The justifications support decision making by keeping track of (source or more generally meta) information. They hence can influence decisions on how to update an agent's knowledge, how to negotiate and what resolutions are acceptable.

Our notion of conflict captures the following concept: Let $\MaxGoalSet$ be the set of maximal goals that $\proc$ beliefs it can achieve with the help of $\procB$. But since $\procB$ might choose a strategy to accomplish some of its maximal goals, $\proc$ believes that it is in a conflict with $\procB$, if it cannot find one winning strategy that fits all possible strategy choices of $\procB$.
\begin{definition}[Believed Possible Conflict]\label{def:conflict}
	Let $\MaxGoalSet_\proc$ be the set of maximal subgoals of
	$\proc$ at state \state for which a
believed winning strategy
$(\strat_\proc,\strat'_\procB):(2^{\props_\proc})^{*} \to
\Act_\proc\times\Act_\procB$ in $\WorldSet_\proc$ exists.

Agent $\proc$ believes at state \state it is in a possible conflict with
$\procB$, if
for each of its winning strategies $(\strat_\proc,\strat'_\procB):(2^{\props_\proc})^{*} \to
\Act_\proc\times\Act_\procB$ for a maximal subgoal
$\GoalSet_\proc\in\MaxGoalSet_\proc$,
\begin{itemize}[nosep,leftmargin=*]
	\item\label{i:bangle}  
	  there is a strategy
		$(\strat'_\proc,\strat_\procB):(2^{\props_\proc})^{*}
		\to \Act_\proc\times\Act_\procB$ and a possible world $\world\in\WorldSet_\proc$
		such that  $(\strat'_\proc,\strat_\procB)$ is a winning strategy in $\world$ for
		$\GoalSet_\procB$, a believed maximal subgoal of 
		the believed goals of  $\procB$ in $\world$.
	      \item\label{i:comb} but 
		$(\strat_\proc,\strat_\procB)$
		is not a winning strategy for $\GoalSet_\proc\cup\GoalSet_\procB$ in $\world_\proc$.
\end{itemize}
\end{definition}
The above notion of conflict captures that $\proc$ analyses the situation within its possible worlds $\WorldSet_\proc(\state)$. It assumes that $\procB$ will follow some winning strategy to accomplish its own goals, while \Env is assumed to behave fully adversarial. $A$ believes that $B$ beliefs that one of A's possible worlds represents the reality.  
It is an interesting future extension to also allow $\proc$ having more  complicated beliefs about the beliefs of $\procB$, as already well supported by the logical framework introduced in \autoref{sec:justifications}. For instance we can capture situations like $A$ considers it possible that there is an obstacle on the road, while it believes that $B$ believes there is no obstacle. This extension does change the base line of our contribution but makes the following presentation more complex. So we refrain from considering beliefs about beliefs for the sake of comprehensibility.

For an example of the conflict notion, consider a situation where $\proc$ drives on a highway side by side of $\procB$ and $\proc$ just wants to stay collision-free, $\proc$ does not believe to be in a conflict situation when it believes that $\procB$ also prioritizes collision-freedom, since $\procB$ will not suddenly choose to crash into $\proc$ which would violate its own goal. But in case $\procB$ has no strategy to accomplish collision-freedom (assume a broken car in front of $\procB$) within $\world_\proc$, then $\proc$ assumes that \procB behaves arbitrarily (achieving its remaining goal $\top$) and $\proc$ believes to be in conflict with \procB.

\subsection{Applying the Formal Notion}
In this subsection we consider the formal notions introduced in the previous subsection and illustrate them -- focusing on the examples given at the start of this section.

\paragraph{Propositional Characterisation of the World} For the sake of a small example, let us consider the following propositional characterisation of a world:
For each agent $X\in\{\proc,\procB\}$ there is a pair of variables $(l_{X},p_X)$ storing its position in the road. Further each agent drives a certain speed $s_X$ abstracted to three different levels, $s_X\in\{0,1,2\}$ encoding slow, medium and fast speed levels.  We consider only time bounded properties. The evolution along the observed time window is captured via copies of $(l_{X},p_X, s_X)$, $(l_{X,t},p_{X,t},s_{X,t})$ where $0\leq t\leq \textit{max\_obs\_time}$ encodes the observed time points.
Each agent $X$ can change lane, encoded by increasing or decreasing $l_X$, and choose between three different speeds, that is, (i) decelerate inducing a change from fast to medium, or, medium to slow, respectively, or (ii) accelerates from slow to medium, or, from medium to fast, respectively. 

\paragraph{A Real World Model} In this setting each state of the real
world model is labelled with propositions $\{l_{X},p_{X},s_{X}\mid
X\in\{\proc,\procB,o\}\}$ and there are transitions from a state $\state$
to a state $\state'$ labelled
\labelE(s,s')=($lc_1, lc_2, lc_3,$ $sc_1, sc_2, sc_3$), where
$lc_i\in\{\texttt{lane\_change},\neg\texttt{lane\_change}\}$ and
$sc_i\in\{a,d,k\}$. 
$lc_i$ encodes whether agent X$_i$ chooses to
perform a lane change and sc$_i$ encodes how X$_i$ chooses to change its speed, i.e., to accelerate, decelerate or to keep its speed. 
The target state is labelled according to effect of the chosen action.

The initial state encodes the start situation (of the tour) and the
subgraph from the initial state to the current state captures the observed past. 
A world model has a branching structure from the current state
towards the future into the possible different options of lane changing
and choices of speed change. Such a world model describes the past, the current state of the world and
possible future evolutions.
For each point in time $t$ there is hence such a world model. See
\autoref{fig:realworld} for a sketch of an example of a world model at a time $t=4$. 

\begin{figure}
	\begin{minipage}{0.35\textwidth}
\includegraphics[width=0.9\textwidth]{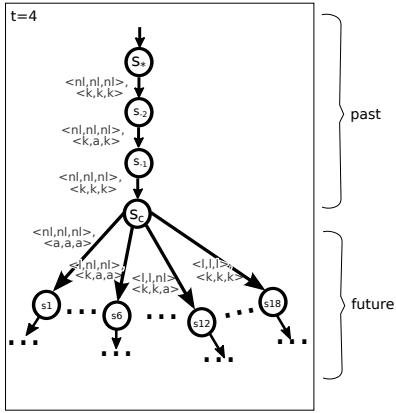}
\end{minipage}
\hfill
\begin{minipage}{0.6\textwidth}
\caption{The transition labelling \labelE\ is sketched within the
  figure itself. The state labelling is omitted there. Let us assume
  that the initial state \labelS($s_*$) is labelled with
  $\{l_\proc=1$, $l_\procB=2$, $l_o=1$, $p_\proc=3$, $p_\procB=1$,
  $p_o=7$, $s_\proc=medium$, $s_\procB=medium$, $s_o=slow\}$ describing
the situation of \autoref{fig:conf}. Currently we are at the time
$t=4$.
\proc, \procB and the environment (determining the moves of the
obstacle) have done three moves: (1) all three stayed at their respective
lane and kept their speed, (2) the same but \procB accelerates and (3)
same as (1). The state labelling reflects the changes induced by the
chosen moves. So the propositions that are true at, e.g., $s_{-2}$
differ from the one of $s_*$ only in terms of the
respective positions:  $\{l_\proc=1$, $l_\procB=2$, $l_o=1$,
$p_\proc=4$, $p_\procB=4$, $p_o=7$, $s_\proc=medium$,
$s_\procB=medium$, $s_o=slow\}$.}\label{fig:realworld}
\end{minipage}
\end{figure}

\paragraph{Possible Worlds} Additional to labelling of states and transitions, the real world is also
labelled with beliefs of the agents at that time. 
The gist is
\begin{enumerate}
  \item[R] the real
world model captures the past, presence and the possible futures at a time $t$. 
\item[B] At time $t$ an agent within world model \world considers a set of worlds
possible. This belief is justified by e.g. evidences from its sensors.
\end{enumerate}
An example is sketched in
\autoref{fig:beliefs}, where only \proc's beliefs are sketched. 
Note that the state labelling, i.e. the set of true propositions,  is not
specified in \autoref{fig:beliefs} in order to declutter the figure.
Some state labelling is given in \autoref{fig:beliefs2}.
\begin{figure}
	\begin{minipage}{0.6\textwidth}
\includegraphics[width=0.9\textwidth]{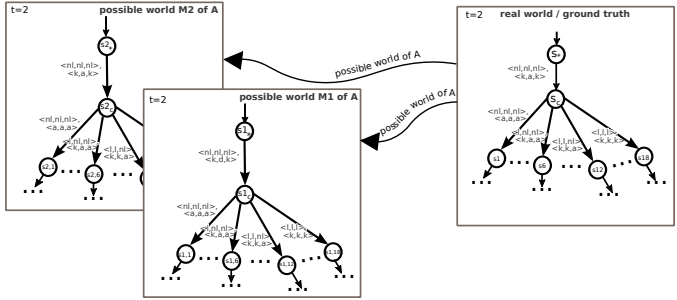}
\end{minipage}
\hfill
\begin{minipage}{0.39\textwidth}
\caption{The real world to the right is associated with beliefs of agent
  \proc. \proc considers at time $t=2$ two worlds as possible, one, M1, bisimilar
  to the real world and a second one, M2, where
\procB accelerates as its first move.}\label{fig:beliefs}
\end{minipage}
\end{figure}
Let  for \autoref{fig:beliefs} be the initial states of the real and the possible worlds be identically labelled, i.e., the agent believes in the "real" past.

Let us now consider  \autoref{ex:evidence}. Agent  \proc has evidence for
\emph{\procB being
fast} and it also has evidence for \emph{\procB being slow}.
\autoref{fig:beliefs2} illustrates that agent \proc considers the two
(sets of)
worlds possible that differ in the
valuation of the respective state propositions.
Agent \proc believes that a world is possible where B is fast -- this is justified by its
radar data--, and \proc considers a world possible where B is
slow -- justified by its lidar. 

We assume that an agent considers any world \world possible that can be
justified by some non-empty set of consistent evidences. 
So a possible world \world satisfies e.g. a set of constraints that is derived from the agent's observations, i.e.,
the sensory evidences, and it also has to be compatible to the agent's laws/rules about
the world, e.g., physical laws.

The evidences provided by radar and lidar in \autoref{ex:evidence}
imply constraints $s_\proc=fast$ and $s_\procB=slow$. 
These constraints are contradictory and hence there is no
possible world that satisfies both constraints. So there
cannot be an arrow in \autoref{fig:beliefs2} from the real world to a
possible world that is labelled with a justification set containing both justifications,
\emph{\Just = $\{$radar, lidar$\}$}. Nevertheless, the radar and lidar evidences
justify that agent \proc believes in alternative worlds (e.g., B is
fast, so it is possible that (a) B was driving at medium speed and accelerated or (b) B was
fast and kept its speed.)
\begin{figure}
	\begin{minipage}{0.6\textwidth}
\includegraphics[width=0.9\textwidth]{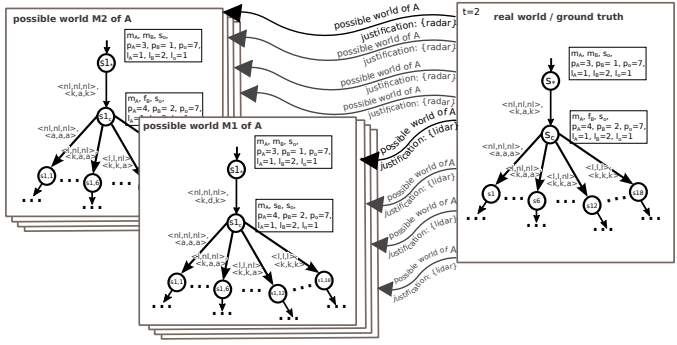}
\end{minipage}
\hfill
\begin{minipage}{0.39\textwidth}
\caption{\proc considers currently two sets of worlds possible, one set
contains all possible worlds where B is slow in accordance to the lidar
 and all worlds in the other set satisfy that B is fast in
accordance to the the radar.}\label{fig:beliefs2}
\end{minipage}
\end{figure}
\paragraph{Strategy and Possible Worlds}
Let us formalise \autoref{ex:winning}. 
\proc considers worlds possible where it has the evidence \texttt{radar:car\,\procB is fast} and hence considers worlds possible where \procB is fast, and also worlds where B is not fast, due to its evidence \texttt{lidar:car\,\procB is slow}. 
We already sketched the possible worlds of \proc above. 
In order to specify the goals of \proc and the goals \proc believes \procB has, we use the usual LTL operators\footnote{\enquote{\Finally} denotes the \textsf{finally} modal operator, \enquote{\Globally} denotes \textsf{globally}, \enquote{\Until} denotes \textsf{until} and in addition we use $\Finally_{\leq t}\varphi$ to express that within the next $t$ steps $\varphi$ has to be true and likewise $\Globally_{\leq t}\varphi$ to specify that at all times up to $t$ $\varphi$ has to hold.}

\proc wants to drive slowly and comfortably, $\varphi_{cf}=\Globally_{\leq 3} (s_\proc=medium\lor s_\proc=slow)$ and avoid collisions $\varphi_{cl}=\Globally_{\leq 10}(p_{\proc}\not=p_{\procB}\land p_{\proc}\not=p_o)$. 
\proc also assumes that \procB wants to avoid collisions,   $\varphi_{\procB,cl}=\Globally_{\leq 10}(p_{\procB}\not=p_{\proc}\land p_{\procB}\not=p_o)$.
The weight assignment to subsets of goals for \proc is
$\goalweight_\proc(\{\varphi_{cl},\varphi_{cf}\})=2$,
$\goalweight_\proc(\{\varphi_{cl}\})=1$ and
$\goalweight_\proc(\Phi)=0$ for all all other subsets $\Phi$. This expresses that collision-freedom is indispensable. Further \proc believes collision-freedom is also indispensable for \procB.
Additionally, \proc derives from $\Phi_\procB$ a constraint that expresses that \procB will not jeopardize collision-freedom and hence it will not drive irrationally into \proc. 
This constraint further restricts the set of worlds that \proc considers possible (cf. \autoref{fig:beliefs3}).  

\begin{figure}
	\begin{minipage}{0.6\textwidth}
\includegraphics[width=0.9\textwidth]{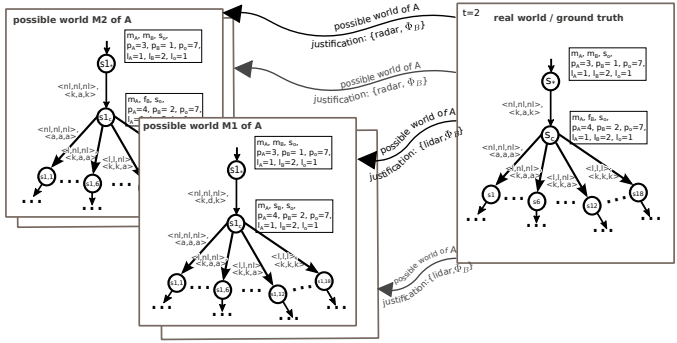}
\end{minipage}
\hfill
\begin{minipage}{0.39\textwidth}
\caption{\proc derived additional constraints from \procB goals that constrain its set of possible worlds.}\label{fig:beliefs3}
\end{minipage}
\end{figure}

In this situation, \proc decides on its next move. 
It is not aware of the state of real world and decides only based on its current beliefs regarding the possible worlds and associated goals of \procB and goal weights. 
\proc determines that staying on its current lane and not changing its speed now is a good move since it can stop and wait in any case, i.e., this move is the prefix of a winning strategy in M1 and all other possible worlds, in which \procB is slow, and also in M2 and all other possible worlds that satisfy that \procB is fast (cf. \autoref{fig:beliefs2}).

\paragraph{Conflicts}
In  \autoref{ex:information} \proc has the goals \begin{itemize*} 
	\item avoid collisions $\varphi_{cl}=\Globally_{\leq 10}(p_{\proc}\not=p_{\procB}$ $\land$ $p_{\proc}\not=p_o)$ and
	\item change lane $\varphi_{lc}=\Finally_{\leq 5}$ \textit{change\_lane} and 
	\item change lane before \procB has passed, if \procB is slow, $\varphi_{flc}=\neg(\Globally_{\leq 3} s_\procB=fast)$ $\Rightarrow$ $p_\proc\geq p_\procB$ $\Until$ $\textit{change\_lane}$ $\land$ $\Finally_{\leq 3}\textit{change\_lane}$ and 
	\item do not change before \procB has passed, if \procB is fast, $\varphi_{slc}=(\Globally_{\leq 3} s_\procB=fast)$ $\Rightarrow$ $\Globally_{\leq 3} \neg\textit{change\_lane}$. 
\end{itemize*} We assume here that \proc has only short term goals and global goals are determined at a higher level.\footnote{Note that also collision-freedom might be sacrificed in so-called dilemma situations.} 
\proc also assumes that \procB wants to avoid collisions.
 The weight assignment to subsets of goals for \proc is specified as follows
 $6=\goalweight_\proc(\{\varphi_{cl},\varphi_{lc},\varphi_{flc},\varphi_{slc}\})$ $>$ $\goalweight_\proc(\{\varphi_{cl},\varphi_{lc},\varphi_{slc}\})$ $=$ $\goalweight_\proc(\{\varphi_{cl},\varphi_{lc},\varphi_{flc}\})$ $>$
 $\goalweight_\proc(\{\varphi_{cl},\varphi_{lc}\})$ $>$
 $\goalweight_\proc(\{\varphi_{cl},\varphi_{flc}\})$ $=$ $\goalweight_\proc(\{\varphi_{cl},\varphi_{slc}\})$\footnote{Note, that $\varphi_{flc}$ does not imply $\varphi_{lc}$.} $>$ $\goalweight_\proc(\{\varphi_{cl},\})=1$,
$\goalweight_\proc(\Phi)=0$ for all all other subsets.

Obviously \proc has a winning strategy  $\strat_{\proc(\procB)}:2^{(2^{\props})^{*}} \rightarrow \Act_\proc\times\Act_\procB$, i.e., if it could determine also \procB's future moves. 
In this case it can achieve $\{\varphi_{cl},\varphi_{lc},\varphi_{flc},\varphi_{slc}\}$. 
If \proc assumes that \procB follows a strategy achieving \procB's own goals under the assumption that \proc will cooperate (i.e. \procB can rule out that \proc changes lane, forcing \procB to decelerate), then \procB can e.g. make up a winning strategy  $\strat_{(\proc)\procB}:2^{(2^{\props})^{*}} \rightarrow \Act_\proc\times\Act_\procB$, where \proc stays at lane 1, \procB at lane 2 and \procB chooses its speed arbitrarily without endangering collision-freedom.  \proc does not have a winning strategy for all these strategies of \procB, since \proc cannot follow the same strategy if (i) \procB is fast in its next three steps and if (ii) \procB is not fast in at least one of the next three steps. 

If \procB tells \proc how fast it will go in its next three steps, the additional information provided by \procB, makes \proc dismiss all possible worlds that do not satisfy the evidence on \procB's future behaviour. \proc can determine an appropriate strategy for all (remaining) possible worlds and the conflict situation is hence resolved.

\section{Epistemic Logic, Justifications and \SystemGraph}\label{subsec:justifications}\label{sec:justifications}
Conflict analysis demands to know who believes to be in conflict with whom and what pieces of information made him belief that he is in conflict. To this end we introduce the \emph{logic of \systemgraphs } that
allows to keep track of external information and extends purely propositional formulae by so called
\emph{belief atoms} (cf. p.~\pageref{p:beliefentities}), which are used to label the sources of
information. In \autoref{sec:conflict} we already used such formulae, e.g.,
"\texttt{radar:car \procB is fast}". 
Our logic provides several atomic accessibility relations representing
justified beliefs of various sources, as required for our examples of \autoref{sec:conflict}.
It provides \systemgraphs as a mean to identify \emph{belief entities} which compose different justifications to consistent information even when the information base contains contradicting information of different sources, as required for analysing conflict situations. 

First, this section provides a short overview on epistemic modal logics
and multi-modal extensions thereof. Such logics use modal operators to 
expressing knowledge and belief stemming from different sources.
Often we will refer to this knowledge and belief as \emph{information},
especially when focusing on the sources or of the information.
Thereupon the basic principles of justifications logics are shortly reviewed.
Justification logics are widely seen as interesting variants to epistemic logics
as they allow to trace back intra-logical and external justifications of derived information.
In the following discussion it turns out that tracing back
external justifications follows the same principles as the distribution
of information over different sources.

Consequently, the concept of information source and external justification are
then unified in our variant of an epistemic modal logic.
This logic of \systemgraphs extends the modal logic by a justification graph. 
The nodes of a justification graph are called belief entities and represent groups of
consistent information. The leaf nodes of a justification graphs are called belief atoms,
which are information source and external justifications at the same time, as they are
the least constituents of external information. We provide a complete axiomatisation
with respect to the semantics of the logic of \systemgraphs.

\subsection{\SystemGraphs}
\myparagraph{Modal Logic and Epistemic Logic}
Modal logic extends the classical logic by modal operators expressing
necessity and possibility. The formula $\nec\phi$ is read as
\enquote{$\phi$ is necessary} and $\pos\phi$ is read as \enquote{$\phi$ is possible}.
The notions of possibility and necessity are dual to each other,
 $\pos\phi$ can be defined as  $\lnot\nec\lnot\phi$.
The weakest modal logic $\System{K}$ extends
propositional logic by the axiom $\Axiom{K}_\nec$ and
the necessitation rule $\Axiom{Nec}_\nec$ as
follows\\
\begin{minipage}{0.45\textwidth}
  \begin{gather*}
    \vdash\nec(\phi\to\psi)\to(\nec\phi\to\nec\psi)\text,\tag{$\Axiom{K}_\nec$}
  \end{gather*}
\end{minipage}
\hfill
\begin{minipage}{0.45\textwidth}
  \begin{gather*}
    \text{from } \vdash \phi\text{ conclude } \vdash \nec\phi\text.
    \tag{$\Axiom{Nec}_\nec$}
  \end{gather*}
\end{minipage}\vspace{1em}\\
The axiom $\Axiom{K}_\nec$ ensures that whenever $\phi\to\psi$ and $\phi$
necessarily hold, then also $\psi$ necessarily has to hold. The necessitation
rule $\Axiom{Nec}_\nec$ allows to infer the necessity of $\phi$ from
any proof of $\phi$ and, hence, pushes any derivable logical truth
into the range of the modal operator $\nec$. This principle is
also known as \emph{logical awareness}.
Various modalities like belief or knowledge can be described by
adding additional axioms encoding the characteristic properties of
the respective modal operator. 
The following two axioms are useful to model knowledge and belief:\\
\begin{minipage}[h]{0.45\textwidth}
  \begin{gather*}
    \vdash \nec\phi\to\phi\text,
    \tag{$\Axiom{T}_\nec$}
  \end{gather*}
\end{minipage}
\hfill
\begin{minipage}[h]{0.45\textwidth}
  \begin{gather*}
    \vdash \nec\phi\to\pos\phi\text.
    \tag{$\Axiom{D}_\nec$}
  \end{gather*}
\end{minipage}\vspace{1em}\\
The axiom $\Axiom{T}_\nec$ and $\Axiom{D}_\nec$ relate necessity with
the factual world. While the truth axiom $\Axiom{T}_\nec$ characterises
\emph{knowledge} as it postulates that everything which is necessary is also factual,
$\Axiom{D}_\nec$ characterises \emph{belief} as it postulates the weaker
property that everything which is necessary is also possible.
Under both axioms $\vdash \nec\bot\to\bot$ holds,
i.e.\ a necessary contradiction yields also a factual contradiction.

Multi-modal logics are easily obtained by adding several modal operators
with possibly different properties and can be used to express the information of more than one agent.
E.g., the formula $e_i\::\phi$ expresses that 
the piece of information $\phi$ belongs to the modality $e_i$. 
Modal operators can also be used to
represent modalities referring to time.
E.g., in the formula $\Next\phi$ the temporal modality $\Next$ expresses
that $\phi$ will hold in the next time step. An important representative of a
temporal extension is linear temporal logic ($\System{LTL}$).

In multi-agent logics the notions of common information
and distributed information play an important role.
While common knowledge captures the information which is known to every agent $\ag_i$,
we are mainly interested in information that is
distributed within a group of agents $E=\{\ag_1,\dots,\ag_n\}$.
The distributed information within
a group $E$ contains any piece of information that
at least one of the agent $\ag_1$, \dots, $\ag_n$ has. Consequently, we
introduce a set-like notion for groups, where an agent $\ag$ is identified
with the singleton group $\{\ag\}$ and the expression
$\{\ag_1,\dots,\ag_n\}\::\phi$ is used to denote that $\phi$ is
distributed information within the group $E$.
The distribution of information is axiomatised by
\begin{gather*}
  \vdash E\::\phi\to F\::\phi\text{, where $E$ is a subgroup of $F$.} \tag{$\Axiom{Dist}_{E,F}$}
\end{gather*}
Note that groups may not be empty.
The \emph{modal logic for distributed information} contains for every
group $E$ at least the axiom $\Axiom{K}_E$, the necessitation rule
$\Axiom{Nec}_E$, and the axiom $\Axiom{Dist}_{E,F}$ for any group $F$ with
$E\subseteq F$.

\myparagraph{Justification Logics}
Justification logics \cite{artemov:jck:2006} are variants
of epistemic modal logics
where the modal operators of knowledge and belief are unfolded into
justification terms. Hence, justification logics allow a complete
realisation of Plato's characterisation of knowledge as justified true belief.
A typical formula of justification logic has the form $s\::\phi$,
where $s$ is a justification term built from justification constants, and it
is read as ``$\phi$ is justified by $s$''.
The basic justification logic $J_0$ results from extending propositional
logic by the application axiom and the sum axioms\\
\begin{minipage}[h]{0.45\textwidth}
  \begin{gather*}
    \vdash s\::(\phi\to\psi)\to(t\::\phi\to[s\cdot t]\::\psi)\text,
    \tag{$\Axiom{Appl}$}
  \end{gather*}
\end{minipage}
\hfill
\begin{minipage}[h]{0.50\textwidth}
  \begin{gather*}
    \vdash s\::\phi\to[s + t]\::\phi\text,\quad
    \vdash s\::\phi\to[t+s]\::\phi\text,
    \tag{$\Axiom{Sum}$}
  \end{gather*}
\end{minipage}\vspace{1em}\\
where $s$, $t$, $[s\cdot t]$, $[s+t]$, and $[t+s]$ are justification terms
which are assembled from justification constants using the operators $+$ and $\cdot$ according to
the axioms.
Justification logics tie the epistemic tradition together
with proof theory. Justification terms are reasonable
abstractions for constructions of proofs. If $s$ is a proof of
$\phi\to\psi$ and $t$ is a proof of $\phi$ then the application
axiom postulates that there is a common proof, namely $s\cdot t$, for $\psi$.
Moreover, if we have a proof $s$ for $\phi$ and some proof $t$ then
 the concatenations of both proofs, $s+t$ and $t+s$,
are still proofs for $\phi$.
In our framework we were not able to derive any meaningful
example using the sum axiom of justification logic. Therefore 
this axiom is omitted in the following discussion.

\myparagraph{Discussion}
All instances of classical logical tautologies,
like $A \lor \lnot A$ and $s\::A \lor \lnot s\::A$, are provable
in justification logics.
But in contrast to modal logics, justification logics do
not have a necessitation rule. The lack of the necessitation rule allows
justification logics to break the principle of logical awareness, as 
$s\::(A \lor \lnot A)$ is not necessarily provable for an arbitrary
justification term $s$.
Certainly, restricting the principle of logical awareness is attractive
to provide a realistic model of restricted logical resources.
Since we are mainly interested in revealing and resolving conflicts, 
the principle of
logical awareness is indispensable in our approach.

Nevertheless, justification logic can simulate unrestricted logical awareness
by adding proper axiom internalisation rules
$\vdash e\::\phi$ for all axioms $\phi$ and justification constants
$e$. In such systems a weak variant
of the necessitation rule of modal logic holds: for any derivation $\vdash \phi$
there exists a justification term $t$ such that $\vdash t\::\phi$ holds.
Since $\phi$ was derived using axioms and rules only,
also the justification term $t$ is exclusively built from justification constants
dedicated to the involved axioms. Beyond that, $t$ is hardly informative
as it does not help to reveal \emph{external} causes of a conflict.
Hence,  we omit
the axiom internalisation rule and add the modal axiom $\Axiom{K}_t$
and the modal necessitation rule $\Axiom{Nec}_t$ for any justification term $t$
to obtain a justification logic where each justification term is closed
under unrestricted logical awareness.

An important consequence of the proposed system is that $\cdot$ becomes
virtually idempotent and commutative.\footnote{%
  For any instance 
  $\vdash s\::(\phi\to\psi)\to(s\::\phi\to[s\cdot s]\::\psi)$ of $\Axiom{Appl}$ there is an instance 
  $\vdash s\::(\phi\to\psi)\to(s\::\phi\to s\::\psi)$ of $\Axiom{K}_s$ in the proposed system. Moreover, it is an easy exercise to show that any instance of 
  $\vdash s\::(\phi\to\psi)\to(t\::\phi\to[t\cdot s]\::\psi)$ is derivable in the proposed system.
}
These insights allows us to argue merely about justification groups instead
of justification terms. It turns out that a proper reformulation
of $\Axiom{Appl}$ with regard to justification groups is equivalent
to $\Axiom{Dist}_{E,F}$, finally yielding the same axiomatisation for
distributed information and compound justifications.

\myparagraph{Belief Atoms, Belief Groups, and Belief Entities}\label{p:beliefentities}
So far, we argued that assembling distributed information and
compound justifications follow the same principle. In the following
we even provide a unified concept for the building blocks of both notions.
A \emph{belief atom} $\atom$ is the least constituent of external information
in our logic. To each $\atom$ we assign the modal operator $\atom\::$. Hence,
for any formula $\phi$ also $\atom\::\phi$ is a formula saying 
\enquote{$\atom$ has information $\phi$}. 
Belief atoms play different roles in our setting.
A belief atom may represent a sensor collecting information about the state
of the  world, or it may represent certain operational rules as well
as a certain goal of the system. The characteristic property of a
belief atom is that the information of a belief atom has to be accepted
or rejected as a whole. Due to its external and indivisible nature,
$e$ is the only source of evidence for its information.
The only justification for information of $e$ is $e$ itself. Consequently,
$\atom\::\phi$ can also be read as \enquote{$\atom$ is the justification for $\phi$}.
This is what belief atoms and justifications have in common:
either we trust a justification or not.

The information of a system is distributed among its belief atoms.
The modal logic for distributed information allows us to
consider the information which is distributed over a \emph{belief group}.
While belief groups can be built arbitrarily from belief atoms, we also
introduce the concept of \emph{belief entities}. A belief entity
is either a belief atom, or a distinguished group of belief entities.
Belief entities are dynamically distinguished by a \systemgraph.
In contrast to belief groups, belief entities and belief atoms are
not allowed to have inconsistent information. Hence a \systemgraph allows us to restrict the awareness of extra-logical evidences -- so we can distinctively integrate logical resources that have to be consistent.

\myparagraph{\SystemGraphs}
\label{sec:section1}
Let $\Vars$ be a set of propositional variables
and let
$\BeliefEntities$ be the set of belief entities. The designated
subset $\BeliefAtoms$ of $\BeliefEntities$ denotes the set of belief atoms.

\begin{definition}[Language of \SystemGraphs]
  A formula $\phi$ is in the language of \systemgraphs 
  if and only if $\phi$ is built according to the following BNF,
  where $A\in\Vars$ and $\emptyset\neq E \subseteq\BeliefEntities$:
  \begin{gather*}
    \phi ::= \bot \mid A \mid (\phi \to \phi) \mid E\::(\phi) \mid \Next(\phi) \mid \Past(\phi) \mid (\phi)\Until(\phi) \mid (\phi)\Since(\phi).
  \end{gather*}
\end{definition}
Using the descending sequence of operator precedences
($\::$, $\lnot$, $\lor$, $\land$, $\to$, $\liff$), we can define the
well-known logical connectives $\lnot$, $\lor$, $\land$ and $\liff$ from $\to$ and $\bot$.
Often, we omit brackets if the formula is still uniquely readable.
We define $\to$ to be right associative.
For singleton sets $\{e\}\subseteq\BeliefEntities$ we also write $e\::\phi$
instead of $\{e\}\::\phi$. The language allows the usage of temporal operators
for \emph{next time} ($\Next$), \emph{previous time} ($\Past$), \emph{until} ($\Until$), and \emph{since} ($\Since$). Operators like \emph{always in the future} ($\Globally$) or
\emph{always in the past} ($\History$) can be defined from the given ones.

\begin{definition}[\SystemGraph]
  A \emph{\systemgraph} is a directed acyclic graph
  $G$ whose nodes are belief entities of $\BeliefEntities$.
  An edge $e\mapsto_G f$ denotes that the belief entity $e$ has the
  \emph{component} $f$.  
  The set of all direct components of an entity $e$ is defined as
  $G(e):=\{ f \mid e \mapsto_G f \}$.

  The leaf nodes of a \systemgraph are populated by belief atoms, i.e.\
  for any belief entity $e$ it holds $e\in\BeliefAtoms$ if and only if
  $G(e)=\emptyset$.
\end{definition}

\begin{definition}[Axioms of a \SystemGraph]\label{def:axioms}
  Let $G$ be a \systemgraph. The logic of a \systemgraph
  has the following axioms and rules.
  \begin{compactenum}[(i)]
  \item As an extension of propositional logic the rule of modus ponens
    $\Axiom{MP}$ has to hold:
    from $\vdash \phi$ and $\vdash \phi\to\psi$ conclude
    $\vdash \psi$. Any substitution instance of
    a propositional tautology $\phi$ is an axiom.
  \item Belief groups are closed under logical consequence and follow
    the principle of logical awareness. Information is
    freely distributed along the subgroup-relation.
    For any belief group $E$ the axiom
    $\Axiom{K}_E$ and the necessitation rule $\Axiom{Nec}_E$ hold.
    For groups $E$ and $F$ with $E\subseteq F$ the
    axiom $\Axiom{Dist}_{E,F}$ holds.
  \item Belief entities are not allowed to have inconsistent information.
    Non-atomic belief entities inherit all information of their
    components.
    For any belief entity $e$ the axiom $\Axiom{D}_e$ holds.
    If $E$ is a subgroup of the components of $e$, then the axiom
    $\Axiom{Dist}_{E,e}$ holds.
    \label{it:inheritance}
  \item In order to express temporal relation the logic for
    the \systemgraph includes the axioms of Past-LTL (LTL
    with past operator). A comprehensive list of axioms can be found
    in \cite{lichtenstein1985}. 
  \item Information of a belief entity
	$e\in\BeliefEntities$ and time are related. The axiom
    $(\Axiom{PR}_E):\quad \vdash e\::\Past\phi \liff \Past e\::\phi$
    ensures that every belief entity $e$ correctly remembers its
    prior beliefs and establishes a principle which is also
    known as \emph{perfect recall} (e.g., see \cite{fagin:reasoning:2003}).
  \end{compactenum}
\end{definition}

\begin{definition}[Proof]
  Let $G$ be a \systemgraph. A proof (derivation) of $\phi$
  in $G$ is a sequence of formulae $\phi_1,\dots,\phi_n$ with $\phi_n=\phi$
  such that each $\phi_i$ is either an axiom of the \systemgraph or
  $\phi_i$ is obtained by applying a rule to previous members
  $\phi_{j_1},\dots,\phi_{j_k}$ with $j_1,\dots,j_k < i$.
  We will write $\vdash_G \phi$ if and only if such a sequence exists. 
\end{definition}

\begin{definition}[Proof from a set of formulae]\label{def:proof from sets}
  Let $G$ be a \systemgraph and $\Sigma$ be a set of formulae.
  The relation $\Sigma \vdash_G \phi$ holds if and only if
  $\vdash_G (\sigma_1 \land \dots \land \sigma_k) \to \phi$ for
  some finite subset $\{\sigma_1,\dots,\sigma_k\}\subseteq\Sigma$ with
  $k\geq 0$.
\end{definition}

\begin{definition}[Consistency with respect to a \systemgraph]\label{def:consistent}
  Let $G$ be a \systemgraph.
  \begin{compactenum}[(i)]
  \item A set $\Sigma$ of formulae is $G$-inconsistent
    if and only if $\Sigma \vdash_G \bot$. Otherwise,
    $\Sigma$ is $G$-consistent.
    A formula $\phi$ is $G$-inconsistent if and only if
    $\{\phi\}$ is $G$-inconsistent. Otherwise, $\phi$ is $G$-consistent.
  \item A set $\Sigma$ of formulae is maximally $G$-consistent
    if and only if $\Sigma$ is $G$-consistent and
    for all $\phi\not\in\Sigma$ the set $\Sigma\cup\{\phi\}$
    is $G$-inconsistent.
  \end{compactenum}
\end{definition}

\myparagraph{Semantics}\label{sec:semantics}
Let $\StateSpace$ be the \emph{state space}, that is the set of all possible
states of the world.
An interpretation $\Interp$ over $\StateSpace$ is a mapping that maps each state
$\state$ to a truth assignment over $\state$,
i.e.\ $\Interp(\state)\subseteq\Vars$ is the subset of all
propositional variables which are true in the state $\state$.
A \emph{run} over $\StateSpace$ is a function $r$ from the natural numbers
(the time domain) to $\StateSpace$. The set of all runs is denoted by $\Points$.

\begin{definition}\label{def:kripke for systems}
  Let $G$ be a \systemgraph. A Kripke structure $M$ for $G$ is a tuple
  $M  = (\StateSpace, \Points, \Interp, ({\mapsto_e})_{e\in\BeliefEntities})$ where
  \begin{compactenum}[(i)]
  \item $\StateSpace$ is a state space, 
  \item $\Points$ is the set of all runs over $\StateSpace$,
  \item $\Interp$ is an interpretation over $\StateSpace$,
  \item each $\mapsto_e$ in $(\mapsto_e)_{e\in\BeliefEntities}$ is an individual accessibility relation
   $\mapsto_e\subseteq\StateSpace\times\StateSpace$ for a belief entity $e$ in $\BeliefEntities$.
  \end{compactenum}
\end{definition}

\begin{definition}[Model for a \SystemGraph]\label{def:model}
  Let $M = (\StateSpace,\Points, \Interp, ({\mapsto_e})_{e\in\BeliefEntities})$ be a
  Kripke structure for the \systemgraph $G$, where 
  \begin{compactenum}[(i)]
  \item $\mapsto_e$ is a serial relation for any belief entity
    $e\in\BeliefEntities$,\label{it:serial}
  \item ${\mapsto_{E}}$ is defined as
    ${\mapsto_{E}} = \bigcap_{e\in E}{\mapsto_e}$ for any belief group
    $E\subseteq\BeliefEntities$,\label{it:distinfo}
  \item ${\mapsto_e} \subseteq {\mapsto_E}$ holds for all
    non-atomic belief entities $e\in\BeliefEntities\setminus\BeliefAtoms$
    and any subgroup $E\subseteq G(e)$.\label{it:inherit}
  \end{compactenum}
  We recursively define the model relation $(M,r(t))\models_G \phi$ as follows:
  \begin{gather*}
    \begin{array}{lcl}
      (M,r(t)) \not\models_G \bot\text{.}
      \\
      (M,r(t)) \models_G Q
      &:\Longleftrightarrow&
      Q\in\Interp(r(t))\text.
      \\
      (M,r(t)) \models_G \phi\to\psi
      &:\Longleftrightarrow&
      (M,r(t)) \models_G \phi \text{ implies }
      (M,r(t)) \models_G \psi\text.
      \\
      (M,r(t)) \models_G E\::\phi
      &:\Longleftrightarrow&
      (M, r'(t)) \models_G \phi
      \text{ for all $r'$ with $r(t')\mapsto_E r'(t')$ for all $t' \leq t$.}
      \\
      (M,r(t)) \models_G \Next\phi
      &:\Longleftrightarrow&
      (M,r(t+1)) \models_G \phi\text{.}
      \\
      (M,r(t)) \models_G \Past\phi
      &:\Longleftrightarrow&
      (M,r(t')) \models_G \phi
      \text{ for some $t'$ with $t'+1=t$.}
      \\
      (M,r(t)) \models_G \phi\Until\psi
      &:\Longleftrightarrow&
      (M,r(t')) \models_G \psi
      \text{ for some $t'\geq t$ and }\\
      &&\text{$(M,r(t'')) \models_G \phi$ for
        all $t''$ with $t\leq t''<t'$.}
      \\
      (M,r(t)) \models_G \phi\Since\psi
      &:\Longleftrightarrow&
      (M,r(t')) \models_G \psi \text{ for some } 0\leq t'\leq t
      \text{ and }\\
      &&\text{$(M,r(t'')) \models_G \phi$ for
        all $t''$ with $t'<t''\leq t$.}
    \end{array}
  \end{gather*}
  When $(M,r(t))\models_G \phi$ holds, we call $(M,r(t))$ a \emph{pointed model}
  of $\phi$ for $G$. If $(M,r(0))$ is a pointed model of $\phi$ for $G$, then
  we write $(M,r)\models_G \phi$ and say that the run $r$ satisfies $\phi$.
  Finally, we say that $\phi$ is satisfiable for $G$, denoted by
  $\models_G \phi$ if and only if there exists a model $M$ and a run $r$
  such that $(M,r)\models_G \phi$ holds.
\end{definition}

\begin{proposition}[Soundness and Completeness]
  The logic of a \systemgraph $G$ is a sound and complete axiomatisation
  with respect to the model relation $\models_G$. That is,
  a formula $\phi$ is $G$-consistent if and only if $\phi$ is satisfiable for
  $G$.
\end{proposition}

While the soundness proof is straightforward, a self-contained
completeness proof involve lengthy sequences of various model
constructions and is far beyond the page limit.
However, it is well-known, (e.g., \cite{gerbrandy1998}),
that $\System{K}^D_n$,
the $n$-agent extension of $\System{K}$ with distributive information
is a sound and complete axiomatisation with respect to the class
of Kripke structures having $n$ arbitrary accessibility relations, where
the additional accessibility relations for groups are given as
the intersection of the participating agents, analogously to
Def. \ref{def:model}.(\ref{it:distinfo}).
Also the additional extension $\System{KD}^D_n$ with
$\Axiom{D}_E$ for any belief group $E$ is sound and complete
with respect to Kripke structures having serial accessibility relations,
analogously to Def. \ref{def:model}.(\ref{it:serial}).
The axioms of \systemgraph are between these two systems.
Def. \ref{def:model}.(\ref{it:inherit}) explicitly allows belief
entities to have more information than its components.
Various completeness proofs for combining LTL and epistemic logics are given
e.g., in \cite{fagin:reasoning:2003}.

\myparagraph{Extracting Justifications}
Let $\Sigma=\{\sigma_1,\dots,\sigma_n\}$ be a finite set of formulae
logically describing the situation which is object of our investigation.
Each formula $\sigma_i\in\Sigma$ encodes information of belief atoms
($\sigma_i\equiv e_i\::\phi_i$ with $e_i\in\BeliefAtoms$), facts
($\sigma_i\equiv \phi_i$ where $\phi_i$ does not contain any
epistemic modal operator), or is an arbitrary Boolean
combinations thereof. Further, let $G$ be a \systemgraph such that
$\Sigma$ is $G$-consistent and $e$ be a non-atomic belief entity of $G$.
For any formula $\phi$ we may now ask whether $\phi$ is part of the
information of $e$. If there is a proof $\Sigma \vdash_G e\::\phi$,
then $\phi$ is included in $e$'s information. To extract
a justification for $e\::\phi$ we use that
$\Sigma\cup\{\lnot e\::\phi\}$ is $G$-inconsistent and accordingly
unsatisfiable for $G$. If we succeed in extracting a minimal unsatisfiable
core $\Sigma' \subseteq \Sigma\cup\{\lnot e\::\phi\}$ a minimal
inconsistency proof can be recovered, from which finally
the used justifications are extracted.

The following proposition allows to use SAT/SMT-solvers for a restricted
setting and has been used in our case study.
\begin{proposition}[SAT Reduction]\label{prop:SAT}
  Let $\Sigma=\{\sigma_1,\dots,\sigma_n\}$ be a set of formulae such
  that each element $\sigma_i$ is of the form $e_i\::\phi_i$
  with $e_i\in\BeliefAtoms$ and $\phi_i$ does not contain any
  epistemic modal operators. Further, let $e$ be an arbitrary belief entity that does not occur in $\Sigma$.
  Then $G = \{ e \mapsto_G e_i | e_i \text{ occurs in } \Sigma\}$ is a \systemgraph for $\Sigma$
  if and only if $\Phi=\{\phi_1,\dots\phi_n\}$ is satisfiable over the
  non-epistemic fragment of the logic of \systemgraphs.
\end{proposition}
\begin{proof}
	The satisfiability relation for the non-epistemic fragment is independent of the accessibility relations
	$\mapsto_e$, $e\in\BeliefEntities$ and, consequently, also independent of $G$. In particular, $\Phi$ is satisfiable if and only if there exists a model $M'=(\StateSpace,\Points, \Interp)$ and a run $r$ such that
	$(M',r)\models\Phi$. 
	
	Let $G=\{ e \mapsto_G e_i | e_i \text{ occurs in } \Sigma\}$ be a graph.
	
	Let us first assume that $G$ is a justification graph for $\Sigma$.
	Then according to Def. \ref{def:model} there exists a Kripke structure $M = (\StateSpace,\Points, \Interp, ({\mapsto_e})_{e\in\BeliefEntities})$ and a run $r$ such that $(M,r)\models_G e_i\::\phi_i$ for all $e_i\::\phi_i\in\Sigma$. Hence, we have $(M,r')\models_G \phi_i$ for all $r'$ with
	$r\mapsto_{e_i} r'$. Furthermore, from item (\ref{it:serial}) and (\ref{it:inherit}) of Def. \ref{def:model} we observe that ${\mapsto_e}$ is not empty and ${\mapsto_e}\subseteq {\mapsto_{e_1}}\cap\dots\cap{\mapsto_{e_n}}$. Hence, there exists at least one run $r'$ that satisfies all formulae in $\Phi$. Since $\Phi$ does not contain epistemic operators, we found a model $M'=(\StateSpace,\Points,\Interp)$ and a run $r'$
	such that $(M',r')\models\Phi$.
	
	For the other direction, let us assume that there exists some model $M'=(\StateSpace,\Points,\Interp)$ and a run $r$
	such that $(M',r)\models\Phi$. We extend $M'$ to a Kripke structure $M=(\StateSpace,\Points, \Interp, ({\mapsto_{e'}})_{e'\in\BeliefEntities})$
	by setting $r\mapsto_{e'} r'$ for all $e'\in\BeliefEntities$ if and only if $r=r'$ for all $r,r'\in\Pi$. Then $(M,r)\models_G e_i\::\phi_i$ for all $e_i\::\phi_i\in\Sigma$ since
	$(M',r') \models \phi_i$ holds for all $r \mapsto_{e_i} r'$ by construction of the accessibility relations.
	Moreover, since all accessibility relations are equal and reflexive, $\mapsto_e$ is serial.
\end{proof}

\section{Identifying and Analysing Conflicts}
\label{sec:section3}
\label{sec:casestudy}
In this section we first present an abstract algorithm for 
the conflict resolution of \autoref{sec:conflict} that starts at level \levelA/ and proceeds resolution stepwise up to level \levelD/. We then sketch our small case study where we applied an implementation of the abstract algorithm. 
\subsection{Analysing Conflicts}\label{sec:glue}
For the analysis of conflicts we employ SMT solvers.
\autoref{prop:SAT} reduces the satisfiability of a justification graph to a SAT problem. To employ SMT solving for conflict analysis, we encode the (real and possible) worlds of \autoref{sec:conflict} via logic formulae as introduced in \autoref{sec:justifications}. 
Each state $\state_i$ is represented
as a conjunction of literals, $\state_i \equiv \bigwedge v\land \bigwedge \neg v'$. 
Introducing a dedicated propositional variable $v_t$ for each $v\in\props$ and time step $t$ allows us to obtain a formula describing a finite run on \world. 
A predicate of the form $\bigwedge_{(\state,\state')\in\Edg} (\state_{t} \rightarrow \state'_{t+1})$ 
encodes the transition relation \Edg. 
The effect of performing an action $a_t$ at state \state is captured by a formula of the form $a_t \to (\state_t \rightarrow \state'_{t+1})$.
Using this we can encode a strategy $\strat$ in a formula $\spred_\strat$ such that its valuations represent runs of $\world$ according to \strat.
All runs according to $\strat$ achieve goals $\Phi$ if and only if $\psi_\strat \land \neg\Phi$ is unsatisfiable.
These logical encodings are the main ingredients for using a SAT solver for our conflict analysis. Since there are only finitely many possible strategies, we examine for each strategy which goals can be (maximally) achieved in a world \world or in a set of worlds \WorldSet. Likewise we check whether \proc has a winning strategy that is compatible with the strategies \proc believes \procB might choose.

Since we iterate over all possible worlds for our conflict analysis, we are
interested in summarising possible worlds. We are usually not interested in all $v_t$ -- e.g. the speed of \procB may at times $t$ be irrelevant. We are hence free to ignore differences in $v_t$ in different possible worlds and are even free to consider all valuations of $v_t$, even if \proc does not consider them  possible. 
This insight leads us to a symbolic representation of the possible worlds,
collecting the relevant constraints. Now the \systemgraph groups the constraints
that are relevant, with other words, $e\::\phi$ and $e'\::\neg\phi$ will not be
components of the same \systemgraph if the valuation of $\phi$ is relevant. In
the following we hence consider the maximal consistent set of possible worlds,
meaning encodings of possible worlds that are uncontradictory wrt.\ the relevant propositions, which are specified via the \systemgraph.

\subsection{Algorithmic approach}
\label{sec:overview}

In this section, we sketch an abstract algorithm for the conflict resolution at
levels \levelA/ to \levelD/ as in \autoref{sec:conflict}. Note that we do not
aim with \autoref{alg:findstrategy} for efficiency or optimal solutions but aim
to illustrate how satisfiability checks can be employed to analyse our
conflicts.

\begin{algorithm}[t]
	\begin{algorithmic}[1]
        \Function{\findstratt/}{$\baseinfoset/, \goalsetAmax/, \goalsetBmax/, \piactA/,
        \piactB/$}
        \State{$\worldsetA/ \gets \Call{\maxcw/}{\baseinfoset/, \piactA/,
	\piactB/}$}\label{f:l1}
        \Comment{construct set of possible worlds}
        \State{$\stratset/_{A} \gets \Call{\stratfunctA/}{\piactA/, \piactB/,
	\worldsetA/, \goalsetAmax/}$}\label{f:l2}
        \Comment{construct $\lb \strategyA/ \mid \pirun{\strategyA/,
        \worldsetA/} \models
        \goalsetA/ \text{ with } \goalsetA/ \in \goalsetAmax/ \rb$}
	\State{$\conflict/ \gets \emptyset$}\label{f:cc}
        \Comment{set of conflict causes}
        \ForAll{$\strategyA/ \in \stratsetA/$ with
        $\pirun{\strategyA/, \worldsetA/} \models \goalsetA/ \in
        \goalsetAmax/$}\label{algline:test}
        \State $\pijustset/ \gets \Call{\testt/}{\strategyA/, \worldsetA/,
	\goalsetA/, \goalsetBmax/, \piactA/, \piactB/}$\label{f:l6}
        \Comment{cf.\ \autoref{alg:test}}
        \If{$\pijustset/ \neq \emptyset$}
        \Comment{\strategyA/ is not winning for all $\worlditemA/ \in
        \worldsetA/$,
        i.e.\ $\pirun{\strategyA/, \worldsetA/} \not\models \goalsetA/$}
	\State{$\conflict/ \gets \conflict/ \cup \lb \pijustset/ \rb$}\label{f:ec}
        \Comment{memorize justifications \pijustset/}
	\State{$\stratsetA/ = \stratsetA/ \setminus \lb \strategyA/ \rb$}\label{f:rmstrat}
        \EndIf
        
        \EndFor
	\If{$\stratsetA/ = \emptyset$}\label{algline:fix}\label{f:strats}
        \Comment{$A$ is in conflict with $B$}
        \For{$i \in [1,2,3,4]$}
        \Comment{traverse resolution levels}
        \State $\baseinfoset/', \goalsetAmax/{'}, \goalsetBmax/{'} \gets
        \Call{\fixt/}{\conflict/, \level/, \baseinfoset/, \goalsetAmax/,
        \goalsetBmax/}$ 
        \Comment{cf.\ \autoref{alg:fix}}
        \If {$(\baseinfoset/' \neq \baseinfoset/) \lor (\goalsetAmax/{'} \neq
        \goalsetAmax/) \lor (\goalsetBmax/{'} \neq
        \goalsetBmax/)$}\label{algoline:fixedpoint}
        \Comment{new information generated}
        \State $\stratsetA/ \gets \Call{\findstratt/}{\baseinfoset/',
        \goalsetAmax/{'}, \goalsetBmax/{'}, \piactA/, \piactB/}$
        \Comment{new attempt with new information}
        \If{$\stratsetA/ \neq \emptyset$}\label{algoline:success}
        \Comment{new attempt was successful, stop and return}
        \State \breakt/
        \EndIf
        \EndIf
        \EndFor
        \EndIf
        \State \Return $\stratsetA/$ \Comment{select $\strategyA/ \in
        \stratsetA/$ to reach some goal in \goalsetAmax/}
    \EndFunction
\end{algorithmic}
    \caption{Determining winning strategy based on observations, goals, and
    possible actions.}
    \label{alg:findstrategy}\label{a:f}
\end{algorithm}

The following algorithms describe how we deal with logic formulae encoding sets
of possible worlds, sets of runs on them, etc.\ to analyse conflicts
(cf.~\autoref{def:conflict}, \autopageref{def:conflict}) via SMT solving.  We
use  $\worlditem/$ to refer to a formula that encodes a maximal consistent set
of possible worlds (cf.~\autoref{sec:glue}), i.e., that corresponds to a
\systemgraph.  We use $\worldset/$ to refer to a set of formulas
$\worlditem/\in\worldset/$ that encode the set of possible worlds $\WorldSet$
structured into sets of possible worlds via \systemgraphs.  We use $\WorldSet$
and $\worldset/$ synonymously.  Also we often do not distinguish between
$\worlditem/$ and $\world$ -- neglecting that $\worlditem/$ represents a set of
worlds that are like $\world$ wrt to the relevant constraints.

\begin{figure}
    \centering
    \includegraphics[scale=0.7]{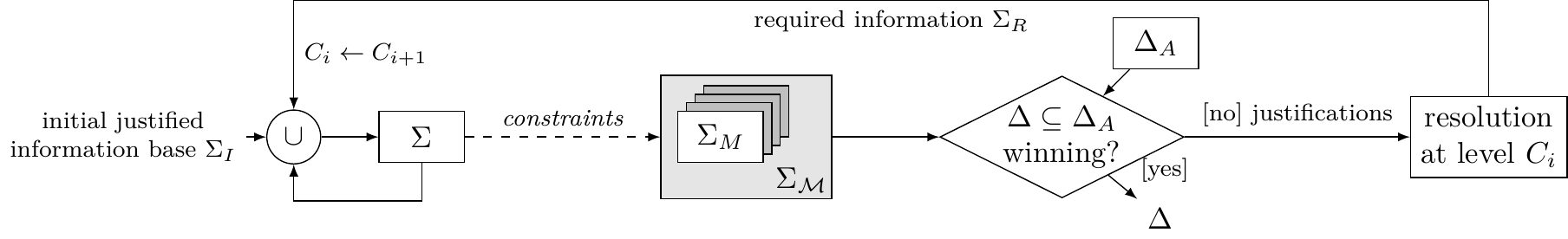}
    \caption{Abstract resolution process with information base, possible
    worlds and strategies.}
    \label{fig:overview}
\end{figure}

\autoref{fig:overview} provides an overview of the relation between the initial
information base $\baseinfoset/_{I}$ of agent $A$, its set \worldsetA/ of
possible worlds \worlditemA/, winning strategies, resolution, and stepwise
update of the information $\baseinfoset/_{R}$ during our conflict resolution
process.  The initial information base defines the set of possible worlds
\WorldSet, which is organised in sets of maximal consistent worlds
$\worlditem/$.  Based on~\worldsetA/, $A$'s set of strategies $\Delta_A$ is
checked whether it comprises a winning strategies in presence of an agent $B$
that tries to achieve its own goals. If no such winning strategy exists, $A$
believes to be in conflict with $B$. At each level \level/ the resolution
procedure tries to determine information $\baseinfoset/_{R}$ of level \level/ to
resolve the conflict. If the possible worlds are enriched by this information,
the considered conflict vanishes. The new information is added to the existing
information base and the over-all process is re-started again until either
winning strategies are found or $\baseinfoset/_{R}$ is empty.

\myparagraph{How to find a believed winning strategy}

\autoref{alg:findstrategy} finds a winning strategy of agent $A$  for a goal
$\goalsetA/$ in $\worldsetA/$ tolerating that $B$ follows an arbitrary winning
strategy in $\worlditem/\in\worldset/$ for its goals, i.e.\ it finds a strategy
that satisfies $\goalsetA/$ in all possible worlds $\worldsetA/$ where
$\goalsetA/$ is maximal for $\worldsetA/$ and in each possible world
$\worlditemA/\in\worldsetA/$ agent $B$ may also follow a winning strategy for
one of its maximal goals $\goalsetB/$. If such a strategy cannot be found, $A$
believes to be in conflict with $B$ (cf.~\autoref{def:conflict}).

Input for the algorithm is (i) a set \baseinfoset/ of formulae describing the
current belief of $A$, e.g.\ its current observations and its history of beliefs
--we call it the \emph{information base} in the sequel--,  (ii) a set of goals
\goalsetAmax/ of $A$ that is maximal in $\WorldSet$, (iii) a set of believed
goals \goalsetBmax/ of $B$ that is maximal for a $\worlditem/ \in \worldsetA/$,
(iv) a set of possible actions \piactA/ for $A$ and  (v) a set of believed
possible actions \piactB/ for~$B$.

First (\lref{f:l1} of \autoref{a:f}) is to construct sets of maximal consistent
sets of possible worlds that together represent \WorldSet. In \lref{f:l2} the
set \stratsetA/ is determined, which is the set of strategies accomplishing a
maximal goal combination for $A$ assuming $B$ agrees to help, i.e., all winning
strategies \strategyA/ that satisfy $\goalsetA/\in\goalsetAmax/$ in all possible
worlds $\worlditemA/ \in \worldsetA/$, where $\goalsetA/$ is maximal for
$\WorldSet$.

\begin{algorithm}[t]
    \begin{algorithmic}[1]
        \Function{\testt/}{$\strategyA/, \worldsetA/, \goalsetA/, \goalsetBmax/,
        \piactA/, \piactB/$}

	\ForAll{$\worlditemA/ \in \worldsetA/$}\label{t:allworlds}
        \State{$\stratset/_{B} \gets \Call{\stratfunctB/}{\piactA/, \piactB/,
	\worlditemA/, \goalsetBmax/}$}\label{t:maxstrat}
        \Comment{construct $\lb \strategyB/ \mid \pirun{\strategyB/,
        \worlditemA/} \models \goalsetB/ \text{ with } \goalsetB/ \in
        \goalsetBmax/ \rb$}
        \ForAll{$\strategyB/ \in \stratsetB/$}
        \ForAll{$\goalsetB/ \in \goalsetBmax/ 
	\text{ with } \pirun{\strategyB/, \worlditemA/} \models \goalsetB/ \text{ and } \goalsetB/\text{ is maximal in }\worlditemA/$}
	\If{$ \pirun{\strategyAB/,\worlditemA/} \not \models \goalsetA/ \cup \goalsetB/$}\label{t:nowin}
        \Comment{$\strategyA/$ is not winning for all $M$ and all $\strategyB/$}
        \State \Return {$\Call{\getjust/}{\strategyAB/
	\not\models \goalsetA/ \cup \goalsetB/}$}\label{t:getJ}
        \Else
        \State \Return $\emptyset$
        \EndIf
        \EndFor
        \EndFor
        \EndFor
        \EndFunction
\end{algorithmic}
    \caption{Test if a strategy is winning in all possible worlds.}
    \label{alg:test}
\end{algorithm}

In lines~\ref*{algline:test}~ff.\ we examine whether one of $A$'s strategies
(where $B$ is willing to help) works even when $B$ follows its strategy to
achieve one of its maximal goals $\goalsetB/ \in \goalsetBmax/$ in \worlditemA/.

To this end \testtt/ is called for all of $A$'s winning strategies
$\strategyA/\in\stratsetA/$ (\lref{f:l6}). The function \testtt/ performs this
test iteratively for one maximal consistent set of worlds~$\worlditemA/$
(\autoref{alg:test} \lref{t:allworlds}). Let $\stratsetB/$  be the set of joint strategies achieving a goal $\goalsetB/ \in
\goalsetBmax/$ that is maximal in $\worlditemA/$.  
We check the compatibility of $A$'s strategy $\strategyA/$ to
every $\strategyB/\in\stratsetB/$ (\autoref{alg:test}~\lref{t:maxstrat}). A strategy of $A$ $\strategyA/$ is
compatible to all of $B$'s strategies $\strategyB/$ if all joint strategies
$\strategyAB/$ achieve the maximal goals for $A$ and $B$
(\autoref{alg:test}~\lref{t:nowin}).{\footnote{Note that according to \autoref{formal}, we have
$\goalsetB/ = \True$ if $B$ cannot achieve any goal.  This reflects that $A$
cannot make any assumption about $B$'s behaviour in such a situation.}}
 If the joint strategy \strategyAB/ is not a
winning strategy for the joint goal $\goalsetA/ \cup \goalsetB/$
(\autoref{alg:test} \lref{t:nowin}), the function \getjustt/ extracts the set of
justifications for this conflict situation (\autoref{alg:test} \lref{t:getJ}).
The set of justifications is added to
the set of conflict causes \conflict/ (\autoref{alg:findstrategy} \lref{f:ec}.). Since strategy \strategyA/ is not
compatible to all of $B$'s strategies, it is hence not further considered as a possible conflict-free strategy for $A$ 
(\autoref{a:f} \lref{f:rmstrat}).

\begin{algorithm}[t]
    \begin{algorithmic}[1]
    \Function{\fixt/}{$\conflict/, \level/, \baseinfoset/, \goalsetAmax/,
    \goalsetBmax/$}
        \For{$\pijustset/ \in \conflict/$}
        \State{$\conflict/ \gets \conflict/ \setminus \lb \pijustset/ \rb$}
        \State{$\baseinfoset/, \goalsetAmax/, \goalsetBmax/ \gets
        \Call{\resolve/}{\baseinfoset/, \pijustset/, \goalsetAmax/,
        \goalsetBmax/, \level/}$}
        \Comment{try resolution according to level \level/}
        \EndFor
        \State \Return $\baseinfoset/, \goalsetAmax/, \goalsetBmax/$ 
    \EndFunction
\end{algorithmic}
    \caption{Try to fix a conflict by resolving contradictions.}
    \label{alg:fix}
\end{algorithm}

A strategy that remains in \stratsetA/ at \autoref{a:f} \lref{f:strats} is a
winning strategy for one of $A$'s goals in all possible worlds~\worlditemA/
regardless of what maximal goals $B$ tries to achieve in~\worlditemA/. However,
if \stratsetA/ is empty at \autoref{a:f} \lref{f:strats} , $A$ is in a
(believed) conflict with $B$ (\autoref{def:conflict}). In this case, conflict
resolution is attempted (cf.~lines~\ref*{algline:fix}~ff.\ in
\autoref{alg:findstrategy}). Function \fixtt/ from \autoref{alg:fix} is called
with the set of conflict causes, the current conflict resolution level, and the
current information base and goals. For each conflict cause, an attempt of
resolution is made by function \resolvet/. The conflict is analysed to identify
whether adding/updating information of the current resolution level helps to
resolve the conflict. If there are several ways to resolve a conflict,
justifications can be used to decide which resolution should be chosen.  Note
that conflict resolution hence means updating of the information base
\baseinfoset/ or goal sets \goalsetAmax/ and \goalsetBmax/. 

Line~\ref*{algoline:fixedpoint} of \autoref{a:f} checks if some new information
was obtained from the resolution procedure. If not, resolution will be restarted
at the next resolution level. If new information was obtained, \findstrattt/ is
called with the updated information. If the result is a non-empty set of
strategies, the algorithm terminates by returning them as (believed) winning
strategies for $A$.  However, if the result is the empty set, resolution is
restarted at the next resolution level. If \stratsetA/ is empty at level
\levelD/, the conflict cannot be resolved and the algorithm terminates.


\myparagraph{Termination}

\autoref{alg:findstrategy} eventually terminates under the following
assumptions. The first assumption is that the set of variables \pivars/ and
hence the information base \baseinfoset/ is finite. In this case, the
construction of maximal consistent possible worlds \worldsetA/ terminates since
there is a finite number of possible consistent combinations of formulae and the
time horizon for the unrolling of a possible world \worlditemA/ is bounded.

Together with finite sets \goalsetAmax/, \goalsetBmax/, \piactA/, and \piactB/,
the construction of strategies, i.e.\ functions \stratfuncttA/ and
\stratfuncttB/, terminates since there are only finite numbers of combinations of
input histories and output action and there is only a finite number of goals to
satisfy.  All loops in algorithms~\ref{alg:findstrategy}, \ref{alg:test}, and
\ref{alg:fix} hence iterate over finite sets.

The extraction of justifications terminates since runs are finite and
consequently the number of actions involved in the run, too. Furthermore, for
each state in a run, there are only a finite number of propositions that apply.
Together with a finite number of propositions representing the goal, \getjustt/
can simply return the (not necessarily minimal) set of justifications from
all these finite many formulae as a naive approach.

\autoref{alg:findstrategy} terminates if a non-empty set \stratsetA/ is derived
by testing and/or resolution, or if a fixed point regarding \baseinfoset/,
\goalsetAmax/, and \goalsetBmax/ is reached. Since all other loops and functions
terminate, the only open aspect is the fixed point whose achievement depends on
\resolvet/. We assume that \autoref{alg:findstrategy} is executed at a fixed
time instance s.t.\ $A$'s perception of the environment does not change during
execution. Thus, \baseinfoset/~contains only a finite number of pieces of
information to share. If we assume that that sharing information leads only to
dismissing possible worlds rather then considering more worlds possible, then
this a monotonic process never removing any information.\footnote{Otherwise the
set of already examined worlds can be used to define a fixed point.}
Furthermore, we assume that the partial order of goals leads, if necessary, to a
monotonic process of goal negotiation which itself can repeated finite many
times until no further goals can be sacrificed or adopted from $B$. Thus, if
\stratsetA/ remains to be the empty set in line~\ref*{algoline:success}, the
fixed point will eventually be reached.

Furthermore, we do not consider any kind of race conditions occurring from
concurrency, e.g.\ deadlock situations where $A$ can't serve $B$'s request
because it does not know what its strategy will be since $A$ wait's for $B$'s
respond, and vice versa.

So in summary, the algorithm terminates under certain artificial assumptions but
cannot determine a resolution in case without an outside arbiter. In practice
such a conflict resolution process has to be equipped with time bounds and
monitors. We consider these aspects as future work.


\subsection{Case study}

We implemented the algorithm sketched above in a Java program employing Yices
\cite{Yices} to determine contradictions and analysed variations of a toy
example to evaluate and illustrate our approach.

We modelled a system of two agents on a two lane highway. Each agent is
represented by its position and its lane. Each agent has a set of actions: it
can change lane and drive forward with different speeds. We captured this via a
discrete transition relation where agents hop from position to position.  The
progress of time is encoded via unrolling, that is we have for each point in
time a corresponding copy of a variable to hold the value of the respective
attribute at that time. Accordingly the transition relation then refers to
these copies.

Since we analyse believed conflicts of an agent, we consider several worlds. In
other words, we consider several variations of a Yices model. Each variation
represents a \systemgraph summarising the maximal consistent set of evidences
and thereby representing a set of worlds which is justified by this set of
evidences. 

We modify the Yices file by adding additional constraints according to the
algorithm \autoref{sec:overview}.  For the steps \levelA/ to \levelD/ we add constraint
predicates, e.g., that encode that information about certain observations have
been communicated by say $B$ to $A$, constraints that specify that $B$ tells $A$
it will decelerate at step $4$ and constraints that encode goal combinations. 

We employed Yices to determine whether there is conflict. The key observation
is: If Yices determines that it holds that $\neg \goalp$ is satisfiable in our
system model, then there is the possibility that the goal is not achieved --
otherwise each evolution satisfies $\goalp$ and there is a winning strategy for
the model.

\section{Related work}
\label{sec:related}
\myparagraph{Studying Traffic Conflicts}
According to Tiwari in his 1998 paper \cite{Tiwari} studying traffic conflicts in India, one of the earliest studies concerned with \emph{traffic conflicts} is the 1963 paper \cite{Harris} of Perkins and Harris. It aims to predict crashes  in road traffic and to obtain a better insight to causal factors.
The term \emph{traffic conflict}  is commonly used according to \cite{Tiwari} as \enquote{an observable situation in which two or more road users approach each other in space and time to such an extent that a collision is imminent if their movements remain unchanged} \cite{Amundsen}. 
In this paper we are interested in a more general and formal notion of conflict. 
We are not only interested in collisions-avoidance but more generally in situations where traffic participants have to cooperate with each other in order to achieve their goals -- which might be collision-freedom. Moreover, we aim to provide a formal framework that allows to explain real world observations as provided by, e.g., the studies of \cite{Tiwari,Harris}.

Tiwari also states in \cite{Tiwari} that it is necessary to develop a better understanding of conflicts and conjectures that \emph{illusion of control} \cite{Langer} and \emph{optimism bias theories} like in \cite{Dejoy} might explain fatal crashes. 
In this paper we develop a formal framework that allows us to analyse conflicts based on beliefs of the involved agents, --although supported by our framework--we here do not compare the real world evolution with the evolution that an agent considers possible. 
Instead we analyse believed conflicts, that are conflicts which an agent expects to occur based on its beliefs. 
Such conflicts will have to be identified and analysed by prediction components of the autonomous vehicles architecture, especially in settings where misperception and, hence, wrong beliefs are possible.

In \cite{SimConflict} Sameh et al. present their approach to modelling conflict resolution as done by humans in order to generate realistic traffic simulations. The trade-off between anticipation and reactivity for conflict resolution is analysed in \cite{ConvexConflict} in order to determine trajectories for vehicles at an intersection.
Both works \cite{SimConflict,ConvexConflict} focus on conflicts leading to accidents. Regarding the suggested resolution approaches, our resolution process suggests cooperation steps with increasing level cooperation. This resolution process is tailored for autonomous vehicles  that remain autonomous during the negation process.   

\myparagraph{Strategies and Games}
For strategy synthesis Finkbeiner and Damm \cite{perimeter} determined the right perimeter of a world model. 
The approach aims to determine the right level of granularity of a world model allowing to find a remorse-free dominant strategy. In order to find a winning (or remorse-free dominant) strategy, the information of some aspects of the world is necessary to make a decision. We accommodated this as an early step in our resolution protocol. Moreover in contrast to \cite{perimeter}, we determine information that agent $A$ then want requests from agent $B$ in order to resolve a conflict with $B$ -- there may still be no winning (or remorse-free dominant) strategy for all goals of $A$. In  \cite{WYRNTKAYN}  Finkbeiner\,et.\,al.presented an approach to synthesise a cooperative strategy among several processes, where the lower prioritised process sacrifices its goals when a process of higher priority achieves its goals. In contrast to \cite{WYRNTKAYN} we do not enforce a priority of agents but leave it open how a conflict is resolved in case not all their goals are achievable. Our resolution process aims to identify the different kinds of conflict as introduced in \autoref{sec:conflict} that arise when local information and beliefs are taken into account and which not necessarily imply that actually goals have to be sacrificed.  

We characterize our conflict notion in a game theoretic setting by considering the environment of agents $A$ and $B$ as adversarial and compare two scenarios where (i) the agent $B$ is cooperative (angelic) with the scenario where (ii)  $B$ is not cooperative and also not antagonistic but reasonable in following a strategy to achieve its own goals. As Brenguier et al. in \cite{Brenguier2017} remark, a fully adversarial environment (including $B$) is usually a bold abstraction.
By assuming in (ii) that $B$ maximises its own goals -- we assume that $B$ follows a winning strategy for its maximal accomplishable goals. So we are in a similar mind set than at assume-guarantee \cite{Henzinger} and assume-admissible \cite{Brenguier2017}  synthesis. Basically we consider the type of strategy (winning/admissible/dominant) as exchangeable, the key aspect of our definition is that goals are not achievable but can be achieved with the help of the other.

\myparagraph{Logics}
Justification logic was introduced in~\cite{artemov:jck:2006,
artemov:logic-justification:2008} as an epistemic logic incorporating knowledge
and belief modalities into justification terms and extends classical modal logic by Plato’s characterisation of knowledge as justified true belief. However,
even this extension might be epistemologically insufficient as Gettier already pointed out in 1963 \cite{gettier1963justified}.
In \cite{artemov2005} a combination of justification logics and epistemic logic
is considered with respect to common knowledge. The knowledge modality $K_i$
of any agent $i$ inherits all information that are justified by some
justification term $t$, i.e.\ $t\::\phi \to K_i\phi$. In such a setting
any justified information is part of common knowledge. Moreover, 
justified common knowledge is obtained by collapsing all justification terms
into one modality $J$ and can be regarded as a special constructive
sort of common knowledge. While our approach neglects the notion of
common information, we use a similar inheritance principle where
a belief entity inherits information of its components, cf.\ Def.\
\ref{def:axioms}.(\ref{it:inheritance}). A comparison of the strength
of this approach with different notions of common knowledge can be found in
\cite{antonakos2007justified}.
While justification logic and related approaches \cite{fagin1987,artemov2009logical}, aim to restrict the principle of logical awareness and the related notion of logical omniscience, we argue in Sec. \ref{subsec:justifications} that
the principle of logical awareness as provided by modal logic is indispensable
in our approach.
A temporal (LTL-based) extension of justification logic has been sketched
in~\cite{bucheli:temporal-justification:2017}. This preliminary work differs
from our approach wrt.\ the axiom systems used for the temporal logic part and
the justification / modal logic part, cf.\ the logic of \systemgraphs
axiomatised in Section~\ref{subsec:justifications}. Our logic and
its axiomatisation incorporates a partial order on the set of beliefs that
underlies their prioritization during conflict resolution, which contrasts with
the probabilistic extension of justification logic outlined
in~\cite{kokkinis:probabilistic-justification:2016}.

\section{Conclusion}
\label{sec:conclusion}
Considering local and incomplete information, we presented a new notion of conflict that captures situations where an agent believes it has to  cooperate with another agent. 
We proposed steps for conflict resolution with increasing level of cooperation.
Key for conflict resolution is the analysis of a conflict, tracing and identifying contradictory evidences. To this end we presented a formal logical framework
unifying justifications with modal logic. Alas, to the authors' best knowledge
there are no efficient satisfiability solvers addressing distributed information
so far. However, we exemplified the applicability of our framework in a
restricted but non-trivial setting. On the one hand, we plan to extend
this framework by efficient implementations of adapted satisfiability solvers,
on the other hand by integrating
richer logics addressing decidable fragments of first order logic,
like linear arithmetic, and probabilistic reasoning.

\bibliographystyle{unsrt}
\bibliography{ms}

\begin{thebibliography}{10}

\bibitem{Yices}
B.~Dutertre.
\newblock Yices 2.2.
\newblock In {\em Computer Aided Verification}, pages 737--744. Springer, 2014.

\bibitem{Galtung}
J.~Galtung.
\newblock {Violence, Peace, and Peace Research}.
\newblock {\em Journal of Peace Research}, 6(3):167--191, September 1969.

\bibitem{artemov:jck:2006}
S.~N. Artemov.
\newblock Justified common knowledge.
\newblock {\em Theor. Comput. Sci.}, 357(1-3):4--22, 2006.

\bibitem{lichtenstein1985}
O.~Lichtenstein, A.~Pnueli, and L.~Zuck.
\newblock The glory of the past.
\newblock In {\em Workshop on Logic of Programs}, pages 196--218. Springer,
  1985.

\bibitem{fagin:reasoning:2003}
R.~Fagin, J.~Y. Halpern, Y.~Moses, and Moshe~Y. Vardi.
\newblock {\em Reasoning About Knowledge}.
\newblock MIT Press, 2003.

\bibitem{gerbrandy1998}
J.~Gerbrandy.
\newblock Distributed knowledge.
\newblock In {\em Twendial 1998: Formal Semantics and Pragmatics of Dialogue},
  volume~98, pages 111--124, 1998.

\bibitem{Tiwari}
G.~Tiwari, D.~Mohan, and J.~Fazio.
\newblock Conflict analysis for prediction of fatal crash locations in mixed
  traffic streams.
\newblock {\em Accident Analysis \& Prevention}, 30(2):207 -- 215, 1998.

\bibitem{Harris}
J.~I. Harris and S.~R. Perkins.
\newblock Traffic conflict characteristics: accident potential at
  intersections.
\newblock {\em Highway Research Board}, 225:35--43, 1967.

\bibitem{Amundsen}
F.~H. Amundsen and C.~Hyden.
\newblock Proc. of the first workshop on traffic conflicts, oslo, norway, 1977.
\newblock 1st Workshop on Traffic Conflicts, LTH Lund.

\bibitem{Langer}
J.~E. Langer.
\newblock The illusion of control.
\newblock {\em Journal of Personality and Social Psychology}, 32:311--328, 08
  1975.

\bibitem{Dejoy}
D.~M. DeJoy.
\newblock The optimism bias and traffic accident risk perception.
\newblock {\em Accident Analysis \& Prevention}, 21(4):333 -- 340, 1989.

\bibitem{SimConflict}
S.~El~hadouaj, A.~Drogoul, and S.~Espi{\'e}.
\newblock How to combine reactivity and anticipation: The case of conflicts
  resolution in a simulated road traffic.
\newblock In {\em Multi-Agent-Based Simulation}, pages 82--96. Springer, 2001.

\bibitem{ConvexConflict}
N.~Murgovski, G.~R. de~Campos, and J.~Sjöberg.
\newblock Convex modeling of conflict resolution at traffic intersections.
\newblock In {\em 2015 54th IEEE Conference on Decision and Control (CDC)},
  pages 4708--4713, Dec 2015.

\bibitem{perimeter}
W.~Damm and B.~Finkbeiner.
\newblock Does it pay to extend the perimeter of a world model?
\newblock In {\em FM 2011: Formal Methods}, pages 12--26. Springer, 2011.

\bibitem{WYRNTKAYN}
W.~Damm, B.~Finkbeiner, and A.~Rakow.
\newblock What you really need to know about your neighbor.
\newblock In {\em Proc. Fifth Workshop on Synthesis, SYNT@CAV 2016}, volume 229
  of {\em {EPTCS}}, pages 21--34, 2016.

\bibitem{Brenguier2017}
R~Brenguier, J.-F. Raskin, and O.~Sankur.
\newblock Assume-admissible synthesis.
\newblock {\em Acta Informatica}, 54(1):41--83, Feb 2017.

\bibitem{Henzinger}
K.~Chatterjee and T.~A. Henzinger.
\newblock Assume-guarantee synthesis.
\newblock In {\em Tools and Algorithms for the Construction and Analysis of
  Systems}, pages 261--275. Springer, 2007.

\bibitem{artemov:logic-justification:2008}
S.~N. Artemov.
\newblock The logic of justification.
\newblock {\em Rew. Symb. Logic}, 1(4):477--513, 2008.

\bibitem{gettier1963justified}
E.~L. Gettier.
\newblock Is justified true belief knowledge?
\newblock {\em Analysis}, 23(6):121--123, 1963.

\bibitem{artemov2005}
S.~Artemov and E.~Nogina.
\newblock Introducing justification into epistemic logic.
\newblock {\em Journal of Logic and Computation}, 15(6):1059--1073, 2005.

\bibitem{antonakos2007justified}
E.~Antonakos.
\newblock Justified and common knowledge: Limited conservativity.
\newblock In {\em International Symposium on Logical Foundations of Computer
  Science}, pages 1--11. Springer, 2007.

\bibitem{fagin1987}
R.~Fagin and J.~Y. Halpern.
\newblock Belief, awareness, and limited reasoning.
\newblock {\em Artificial intelligence}, 34(1):39--76, 1987.

\bibitem{artemov2009logical}
S.~Artemov and R.~Kuznets.
\newblock Logical omniscience as a computational complexity problem.
\newblock In {\em Proc. of the 12th Conference on Theoretical Aspects of
  Rationality and Knowledge}, pages 14--23. ACM, 2009.

\bibitem{bucheli:temporal-justification:2017}
S.~Bucheli, M.~Ghari, and T.~Studer.
\newblock Temporal justification logic.
\newblock In {\em {\rm Proc. of the 9th Workshop on} Methods for Modalities,
  {\rm January 2017}}, volume 243 of {\em EPTCS}, pages 59--74. Open Publishing
  Association, 2017.

\bibitem{kokkinis:probabilistic-justification:2016}
I.~Kokkinis, Z.~Ognjanovic, and T.~Studer.
\newblock Probabilistic justification logic.
\newblock In {\em Logical Foundations of Computer Science - International
  Symposium, {LFCS} 2016. Proc.}, volume 9537 of {\em LNCS}, pages 174--186.
  Springer, 2016.

\end{thebibliography}

\end{document}